\documentclass[lettersize,journal]{IEEEtran}

\usepackage[table,xcdraw]{xcolor}
\usepackage{changes}
\usepackage{amssymb}
\usepackage[T1]{fontenc}
\usepackage[utf8]{inputenc}
\usepackage{graphicx}
\usepackage{array}
\usepackage{listings}
\usepackage[most]{tcolorbox}
\usetikzlibrary{arrows.meta}
\usepackage{tikz}
\usetikzlibrary{positioning, arrows.meta, fit, backgrounds}
\usetikzlibrary{arrows.meta, positioning, shapes.multipart}
\usepackage[]{subcaption}
\usepackage{caption}

\usepackage{enumitem}

\usepackage{url}
\usepackage{titlesec}
\usepackage{pifont}

\definecolor{lightblue}{RGB}{173, 216, 230}
\definecolor{khaki}{RGB}{240, 230, 140}
\definecolor{lightgreen}{RGB}{144, 238, 144}

\usepackage{graphicx}
\usepackage{array}
\usepackage{blkarray}
\usepackage{multirow}
\usepackage[]{subcaption}
\usepackage{caption}

\usepackage{amsthm}
\newtheorem{lemma}{Lemma}
\newtheorem{theorem}{Theorem}
\newtheorem{problem}{Problem}

\usepackage{algorithm}
\usepackage{algpseudocode}

\newtheorem{definition}{Definition}

\usepackage{booktabs,tabularx}
\usepackage{makecell}
\usepackage{array}

\title{\LARGE \bf
	Constraint-Driven Synthesis of Hyper Petri Nets
}

\author{Maksym Figat$^{a}$,
	Alessandro Pinto$^{b}$
	\thanks{$^{a}$Maksym Figat (corresponding author) is with the Warsaw University of Technology (WUT), Warsaw, Poland. \texttt{maksym.figat@pw.edu.pl}. This work was initiated while he was a Visiting Postdoctoral Researcher at the NASA Jet Propulsion Laboratory (JPL), California Institute of Technology, Pasadena, CA, USA.
	}
	\thanks{$^{b}$The author is with NASA Jet Propulsion Laboratory (JPL), California Institute of Technology, Pasadena, CA, USA.}
	\thanks{This work has been submitted to the IEEE for possible publication. Copyright may be transferred without notice, after which this version may no longer be accessible.}
	}

\begin{document}
	\maketitle
	\thispagestyle{empty}
	\pagestyle{empty}
	
	\begin{abstract}
		This paper addresses the modeling and synthesis of constrained robotic system behaviors using Petri nets (PNs). It investigates how to construct models in which all observable system states satisfy given logical constraints while remaining consistent with executable transition semantics. To answer this, we introduce the Hyper Petri Net (HyPN) approach, which synthesizes Petri nets from Boolean specifications while explicitly distinguishing between observable markings and underlying Petri net execution. The proposed method introduces an explicit execution semantics over observable states, induced by admissible (atomic) firing sequences, ensuring by construction that all observable markings satisfy the constraints and revealing a fundamental mismatch between logical feasibility and executable behavior. This is demonstrated in two scenarios inspired by a lunar rover system. These results are particularly relevant for the design of robotic and autonomous systems, as they provide a structured way to ensure correct system configurations while explicitly accounting for execution constraints. The proposed framework further suggests new research directions in execution abstraction, admissible transition systems, and policy selection for navigating between constraint-satisfying states.
	\end{abstract}
	
	\def\abstractname{Note to Practitioners}
	\begin{abstract}
		Robotic and automation systems are commonly developed by translating operational and safety requirements into supervisory logic and then verifying that the resulting controller satisfies the intended constraints. As the number of interacting requirements increases, this process becomes difficult because engineers must manually encode mutual exclusions, resource constraints, and other operational rules, while exhaustive verification is often impractical due to state-space explosion. This paper presents a method for automatically synthesizing Hyper Petri Nets directly from Boolean requirements. The synthesized models guarantee by construction that all observable system states satisfy the specified constraints, while explicitly showing which transitions between these states are executable and which are not. This distinction is important in safety-critical applications, where a logically correct configuration may still be unreachable during execution. The approach was demonstrated using two supervisory-control scenarios inspired by a lunar rover mission. Because the resulting models are standard Petri nets, they can be further analyzed using established techniques such as reachability analysis, deadlock detection, and structural verification. The synthesis procedure is modular, and the size of the generated model grows linearly with the number of variables and constraints, allowing compact models to represent large sets of admissible configurations. The same representation can also be used during operation to monitor which requirements are currently satisfied and which are violated. The method reduces manual modeling effort and increases confidence in system correctness, making it relevant to robotics, industrial automation, and other safety-critical systems, including applications requiring transparent supervision of AI-based decision-making. The current framework focuses on Boolean invariant constraints. Future work will extend the approach to temporal and quantitative requirements, policy selection among alternative execution paths, and integration with higher-level decision-making methods.
	\end{abstract}
	


\section{Introduction}
\label{sec:introduction}

\IEEEPARstart{I}{n} many robotic applications, and especially in safety-critical domains such as space exploration, safety is often more important than performance. \emph{Invariant properties} represent an important and broad class of safety requirements. Informally, invariants are conditions expressed over the state space that must hold at all times during system execution. Examples include maintaining resource levels above critical thresholds, or ensuring that mutually incompatible actions are never executed simultaneously. Traditionally, ensuring that a robotic system satisfies such invariants places a significant burden on engineers. In practice, they must tackle two difficult tasks: (1) they must design a control policy that satisfies the invariants over the set of all reachable states, and (2) they must verify that the resulting closed-loop system indeed respects those invariants. Both steps are challenging and error-prone. To reason about such systems, it is necessary to distinguish between logically admissible states and how they are connected through execution. This is captured by an explicit execution semantics defining transitions between states.

Designing a policy under many interacting safety constraints is difficult because the constraints may be numerous and coupled. In effect,  engineers are often forced to manually encode a complex decision procedure that resembles a specialized constraint solver. As the number of constraints grows, this process becomes increasingly ad hoc, brittle, and difficult to scale. Small changes in the specification such as the addition of an invariant may require substantial redesign of the policy logic. Verification is also problematic. In realistic robotic systems, the set of reachable states can be extremely large. As a result, exhaustive exploration is impractical. Instead, practitioners rely heavily on simulation and testing which is incomplete: it only explores a small subset of the possible executions, and the quality of the evidence depends strongly on the scenarios selected by subject-matter experts. 

The gap between stringent safety requirements and limited assurance guarantees at the implementation level motivates a synthesis approach. A synthesis tool starts from the set of invariant properties and synthesizes an executable behavioral model whose admissible runs satisfy the given invariants by construction. In this way, safety is not treated as a property to be tested after the fact, but as a structural property of the synthesized behavior. Many executable models have been used in robotics such as state machines, behavior trees \cite{ghzouli2023behavior,iovino2022survey}, and Petri Nets \cite{Figat:2022:RAS}. We are interested in Petri Nets (PNs) for several reasons including their ability to naturally capture resource limitations, and the support for analysis of several properties including deadlocks, invariants, and safety. In this paper, \emph{we propose a scalable synthesis process for generating a special class of Petri Nets from a set of invariant constraints}. 

\subsection{Related work}
\label{subsec:related-work}
Model-driven engineering frameworks for robotics focus on integrating design-time abstractions with system implementation, improving development efficiency but often relying on predefined architectural patterns~\cite{Brugali:2020}. Other works emphasize formal analysis and verification of robotic systems using Petri nets (PNs) and related formalisms, highlighting challenges related to state-space explosion and model complexity~\cite{Silva:2021,DalZilio:2023}. 

Verification suffers from scalability issues when applied to complex systems. To address scalability, hierarchical PN models have been proposed, introducing structured representations across multiple abstraction levels~\cite{Huber91}. Extensions of these models consider compositional and behavioral semantics~\cite{Vogler:1992,Zuberek:1996}, as well as applications to complex and distributed systems~\cite{Luo:2015}. However, unrestricted composition mechanisms often complicate analysis and property preservation.

The Robotic System Hierarchical Petri Net (RSHPN) meta-model~\cite{Figat:2022:RAS} introduces a structured hierarchical representation of robotic systems across agent, subsystem, and behaviour levels, with parameters defined using the Robotic System Specification Language (RSSL)~\cite{Figat:2022:RAL}. This approach enables modular modeling and supports automated controller generation. Subsequent work~\cite{Figat:2026:analysis} demonstrates how structural decomposition and property inheritance can improve scalability of formal analysis. The $\rm RS(TM)^2$ methodology~\cite{Figat:2026:RSTM2} further extends this line of research by synthesizing system architectures from high-level ontological descriptions and enabling simulation-based evaluation across multiple abstraction levels.

The problem of synthesizing system behaviors from high-level specifications has been studied in formal methods and control theory. In particular, the synthesis of controllers from temporal logic specifications has been extensively investigated, with applications in robotics and autonomous systems~\cite{wongpiromsarn2012receding, shamgah2018path, maoz2020just,tadewos2022specification}. However, these approaches typically focus on the expressiveness of the specification language and often face scalability challenges in complex systems. Temporal logic specifications can capture both safety and liveness properties, but the resulting synthesis problem can be computationally intractable for large state spaces or complex specifications.

In the context of PNs, \cite{lacerda2019petri} proposes a method for synthesizing PNs from temporal logic specifications. Starting from a PN model of a system to be controlled, and a set of linear temporal logic constraints, their method synthesizes a supervisor enforcing the constraints. The approach has doubly exponential complexity in the size of the constraint formula, due to its translation into a deterministic finite automaton. A related line of work focuses on synthesizing PN supervisors from Generalized Mutual Exclusion Constraints (GMECs)~\cite{iordacheSupervisionBasedPlace2006,Badouel2015, li2010synthesis,iordache2005survey}. In these approaches, the problem is defined by a PN and linear inequalities over markings. The synthesis procedure adds places and transitions to enforce constraints.

The work presented in this paper also focuses on invariant properties. However, we don't assume a PN as a starting point for the definition of invariants. The problem we solve takes as input a set of propositional formulas, and directly generates a PN model whose observable reachable markings satisfy the given formulas. To achieve this goal in a scalable way, we introduce a new class of PNs called Hyper Petri Nets (HyPNs) that allow for the development of a compositional synthesis process. HyPNs are related to the work by Mukund on transition system models of concurrency \cite{mukund1992transition}, and the work of Fujita on hyperautomata \cite{Fujita:2025:hyperautomata}, but differ in their focus on the distinction between observable states and execution semantics.

The HyPN synthesis framework proposed in this paper represents a new design point compared to the related work. First, it accepts a logical specification as input that focuses on safety properties captured by invariants. This is not as expressive as temporal logic, but it is sufficient and expressive enough for many safety-critical applications. Second, it is compositional and scalable, meaning that adding a new invariant reduces to generating a new HyPN for that invariant and composing it with the existing model. Furthermore, the size of the HyPN does not grow exponentially. 

\subsection{Contribution}
\label{subsec:contribution}
This paper makes the following contributions:
\begin{itemize}
	\item Introduction of the HyPN model, which separates observable system states from underlying PN execution and defines an explicit execution semantics over them.
	\item Development of a constraint-driven synthesis procedure that constructs HyPN models from Boolean specifications in conjunctive normal form, ensuring that all observable markings satisfy the given constraints by construction.
	\item Formalization of execution via admissible (atomic) firing sequences, revealing a gap between constraint satisfaction and realizable system behavior, reflected in restricted transitions with respect to reachability.
	\item Validation of the proposed approach on a lunar rover-inspired case study, illustrating constraint-driven synthesis and execution semantics.
\end{itemize}

\subsection{Paper outline}
\label{subsec:paper-outline}
The remainder of the paper is organized as follows. Section~\ref{sec:hypn} introduces the Hyper Petri Net (HyPN) model and distinguishes between observable states and PN execution. Section~\ref{sec:model-synthesis} presents the synthesis procedure, and Section~\ref{sec:formal-properties} analyzes its properties and scalability. Section~\ref{sec:endurance-case-study} describes the lunar rover case study, while Section~\ref{sec:scenarios} presents representative scenarios. Section~\ref{sec:discussion} discusses implications for system design, and Section~\ref{sec:conclusions} concludes the paper.

\section{Hyper Petri net (HyPN)}
\label{sec:hypn}

We consider a classical weighted Petri net (PN)~\cite{Murata:1989:Petri:Nets} defined as $PN = (P, T, F, W, M_0)$, where $P$ and $T$ are finite sets of places and transitions, $F \subseteq (P \times T) \cup (T \times P)$ is the set of arcs, and $W: F \to \mathbb{N}^{+}$ assigns weights, and $M_0$ the initial marking. A \emph{marking} is a function $M: P \to \mathbb{N}$. For a transition $t \in T$, ${}^{\bullet}t$ and $t^{\bullet}$ denote its sets of input and output places, respectively. A transition $t \in T$ is enabled at marking $M$ if $\forall p \in {}^{\bullet}t:\; M(p) \ge W(p,t)$, and its firing produces a new marking. The set of reachable markings from $M_0$ is $\mathcal{R}(M_0)$.

In standard PNs, each reachable marking is a valid system state. This interpretation leads to two problems. First, modelers and analysts are forced to write properties considering all system states and the possible transitions among them, even if they are not interested in certain implementation details (such as the sequence of intermediate states reached during initialization). Second, when all states are considered as valid, enforcing constraints typically requires additional structural mechanisms, such as guard places or mutual-exclusion constraints. As system complexity grows, these mechanisms may significantly increase model size or overly restrict system behavior, potentially leading to deadlocks.

The primary concern in analysis should be to ensure that constraints hold at selected, semantically meaningful states. This perspective is related to the concept of Hyper Automata~\cite{Fujita:2025:hyperautomata}, where transitions are defined over higher-order structures rather than individual steps. Similarly, we argue that certain sequences of elementary actions should be treated as atomic, without exposing intermediate configurations.

To address this limitation, we introduce the Hyper Petri Net (HyPN) model. This motivates a formal distinction between \textbf{internal markings}, produced during execution, and \textbf{observable markings}, which represent stable and semantically meaningful system states. Instead of interpreting the firing of each transition as a state change, the system evolves between observable markings through finite sequences of transitions treated as atomic execution steps, and referred to as \emph{hyper-transitions}, which induce state changes only at their endpoints.

\subsection{Definition and Semantics}

\begin{definition}[Hyper Petri Net]
A \emph{Hyper Petri Net} (HyPN) is a tuple:
\begin{equation}
	\label{eq:hypn}
	HyPN = (P, T, F, W, M_0, \mathcal{M}^{obs}),
\end{equation}
where $(P, T, F, W, M_0)$ is a standard PN and $\mathcal{M}^{obs} \subseteq \mathcal{R}(M_0)$ denotes the set of \emph{observable markings}.
\end{definition}

We define the execution semantics of a HyPN in terms of its induced abstract transition system, which captures the possible behaviors as sequences of observable markings according to the semantics of a non-deterministic automaton.

In the following, the terms \emph{atomic observable transition} and \emph{macro-step} are used interchangeably to denote a single abstract transition between two observable markings.

As in standard PNs, a transition $t \in T$ is enabled at marking $M$, denoted $enabled(t,M)$, if $\forall p \in {}^{\bullet}t:\; M(p) \ge W(p,t)$.
The firing of $t$ at $M$ produces a successor marking $M' = fire(t,M)$. This is also represented as $M \xrightarrow{t} M'$. We refer to transitions that are enabled according to this standard definition as \textbf{\emph{locally enabled}}. Also, let $M_i$ be a marking and $\sigma = \langle t^{(1)}, \dots, t^{(k)} \rangle \in T^*$ a sequence of transitions. Then, $\sigma$ is \textbf{\emph{realizable}} from $M_i$, denoted $realizable(M_i,\sigma)$, if there exists a marking sequence:
\[
M_i = M_{i,0}
\xrightarrow{t^{(1)}} M_{i,1}
\xrightarrow{t^{(2)}} \dots
\xrightarrow{t^{(k)}} M_{i,k}
\]
such that for all $n \in \{1,\dots,k\}$, $enabled\big(t^{(n)}, M_{i,n-1}\big)$. We say that the firing of a realizable sequence $\sigma$ from $M_i$ produces the sequence of markings $(M_{i,0}, \dots, M_{i,k}) = fire(\sigma, M_i)$. We use the shorthand notation $last(\sigma, M_i) = fire(\sigma, M_i)(k)$ to denote the last marking resulting from the execution of the entire sequence.

\begin{definition}[Sequences of transitions]
	Let $M_i$ be a marking and $\sigma = \langle t^{(1)}, \dots, t^{(k)} \rangle \in T^*$ a realization sequence. Then:
	\begin{itemize}
		\item $\sigma$ is \textbf{\emph{admissible}} from $M_i$, denoted $admissible(M_i,\sigma)$, if and only if $realizable(M_i,\sigma)$ holds, and $last(\sigma, M_i) \in \mathcal{M}^{obs}$. 
		\item $\sigma$ is \textbf{\emph{atomic}} from $M_i$, denoted $atomic(M_i,\sigma)$,  if and only if $M_i \in \mathcal{M}^{obs}$, $last(\sigma,M_i) \in \mathcal{M}^{obs}$, and $\forall \ell \in \{1,\dots,k-1\} : fire(\sigma,M_i)(\ell) \notin \mathcal{M}^{obs}$, i.e., only the endpoints are observable.
		\item A transition $t \in T$ is hyper-enabled at $M_i \in \mathcal{M}^{obs}$, denoted $enabled_H(t,M_i)$ if and only if $\exists \sigma \in T^* : atomic(M_i,\sigma) \land \sigma(1)=t$.
	\end{itemize}
\end{definition}

Given two observable markings $M_i, M_j \in \mathcal{M}^{obs}$, we define the \textbf{\emph{atomic sequence set}} $\Sigma_H(M_i,M_j) = \{\sigma \in T^* \mid atomic(M_i,\sigma) \land last(\sigma,M_i) = M_j\}$ as the set of atomic sequences leading from $M_i$ to $M_j$. The \textbf{\emph{atomic sequence sets}} can be computed using the algorithm in Sec.~\ref{subsec:algorithm_sigma_h}. 

\begin{definition}[Abstract state transition system]
Let $HyPN = (P,T,F,W,M_0,\mathcal M^{obs})$ be a Hyper Petri Net, and $\Sigma_H$ its atomic sequence sets. Then, the abstract state transition system of $HyPN$ is $\mathcal{S}_{HyPN} = (\mathcal{M}^{obs}, \Rightarrow, M_0)$ where $M_0 \in \mathcal{M}^{obs}$ is the initial state and $\Rightarrow \subseteq \mathcal{M}^{obs} \times \mathcal{M}^{obs}$ is defined as follows:
\[
M_i \Rightarrow M_j \iff \Sigma_H(M_i,M_j) \neq \emptyset.
\]
\end{definition}
Note that $\Rightarrow$ captures only direct transitions and is not transitive in general, since atomic sequences cannot pass through intermediate observable markings.

\subsection{Computation of $\Sigma_H$}
\label{subsec:algorithm_sigma_h}
Given a HyPN structure, initial marking $M_0$, and observable markings $\mathcal{M}^{obs}$, $\Sigma_H$ is computed according to Algorithm~\ref{alg:sigma_h}. The exploration uses a stack $\mathcal{Q}$ of tuples $(M,\sigma,Visited)$, where $M$ is the current marking, $\sigma$ the firing sequence to $M$, and $Visited$ is the set of visited markings along the current branch. Here, $\epsilon$ denotes the empty firing sequence, and $\sigma \circ t$ appends transition $t$ to $\sigma$.

\begin{algorithm}
	\caption{Computation of $\Sigma_H$}
	\label{alg:sigma_h}
	\begin{algorithmic}[1]
		\For{each $M_i \in \mathcal{M}^{obs}$}
		\State Initialize $\Sigma_H(M_i,M_j) \gets \emptyset$ for all $M_j \in \mathcal{M}^{obs}$
		\State $\mathcal{Q} \gets \{(M_i,\epsilon,\{M_i\})\}$
		
		\While{$\mathcal{Q} \neq \emptyset$}
		\State Pop $(M,\sigma,Visited)$ from $\mathcal{Q}$
		
		\For{each $t \in T$ with $enabled(t,M)$}
		\State $M' \gets fire(M,t)$
		\State $\sigma' \gets \sigma \circ t$
		
		\If{$M' \in \mathcal{M}^{obs}$ and $M' \neq M_i$}
		\State $\Sigma_H(M_i,M') \gets \Sigma_H(M_i,M') \cup \{\sigma'\}$
		\ElsIf{$M' \notin Visited$}
		\State Push $(M',\sigma',Visited \cup \{M'\})$ onto $\mathcal{Q}$
		\EndIf
		\EndFor
		\EndWhile
		\EndFor
	\end{algorithmic}
\end{algorithm}

The search enumerates only simple firing paths between observable markings, i.e., markings are not revisited within the same branch, preventing cyclic sequences from being included in $\Sigma_H$. The exploration is performed independently for each $M_i \in \mathcal{M}^{obs}$, since $\Sigma_H$ is defined as a relation parameterized by observable markings. Consequently, the computation can be parallelized across different initial observable markings $M_i$.

\subsection{Example}
To illustrate the execution semantics of HyPN, consider the example depicted in Fig.~\ref{fig:hypn_example}. The net consists of two independent components, which allows the enabled transitions to be executed in arbitrary order. Let $\mathcal{M}^{obs} = \{ M_0, M_1 \}$ denote the set of observable markings. The marking $M_0 = (1,0,1,0,1,1)$ represents the initial observable configuration, whereas $M_1 = (0,1,0,1,1,1)$ represents the target observable configuration. All remaining reachable markings are treated as internal, i.e., they are excluded from $\mathcal{M}^{obs}$. At $M_0$, two transitions $t_1$ and $t_2$ are locally enabled. Firing either transition produces an intermediate marking that remains reachable in the underlying PN but is not included in $\mathcal{M}^{obs}$. The target observable marking $M_1$ is obtained only after both subsystem updates are completed. The observable transition $M_0 \Rightarrow M_1$ admits multiple atomic fire sequences, i.e., $\Sigma_H(M_0,M_1) = \{ \langle t_1,t_2\rangle, \langle t_2,t_1\rangle \}$. Both sequences are realizable and atomic, and therefore induce the same observable transition $M_0 \Rightarrow M_1$. They are explicitly given by: $\sigma^{(1)} = \langle t_1,t_2\rangle : \quad M_0 \xrightarrow{t_1} M_{0,1} \xrightarrow{t_2} M_1$, and $\sigma^{(2)} = \langle t_2,t_1\rangle : \quad M_0 \xrightarrow{t_2} M'_{0,1} \xrightarrow{t_1} M_1.$ 

The intermediate markings $M_{0,1}$ and $M'_{0,1}$ are reachable in the PN but excluded from $\mathcal{M}^{obs}$ by design. Hence, they are internal and do not appear in the transition system $\mathcal{S}_{HyPN}$, with $\mathcal{M}^{obs} = \{M_0, M_1\}$ and initial state $M_0$. By symmetry, the same reasoning applies in the reverse direction. From $M_1$, the locally enabled transitions $\{t_3,t_4\}$ generate internal markings leading to $M_0 \in \mathcal{M}^{obs}$. Thus, $\Sigma_H(M_1,M_0) = \{\langle t_3,t_4\rangle, \langle t_4,t_3\rangle\}$, implying $M_0 \Rightarrow M_1$ and $M_1 \Rightarrow M_0$.

This example shows that HyPN preserves the full behavior of the underlying PN while enforcing atomicity at the level of observable markings. Multiple admissible internal firing sequences may correspond to a single observable transition.

\begin{figure}
	\centering
	\includegraphics[width=0.85\linewidth]{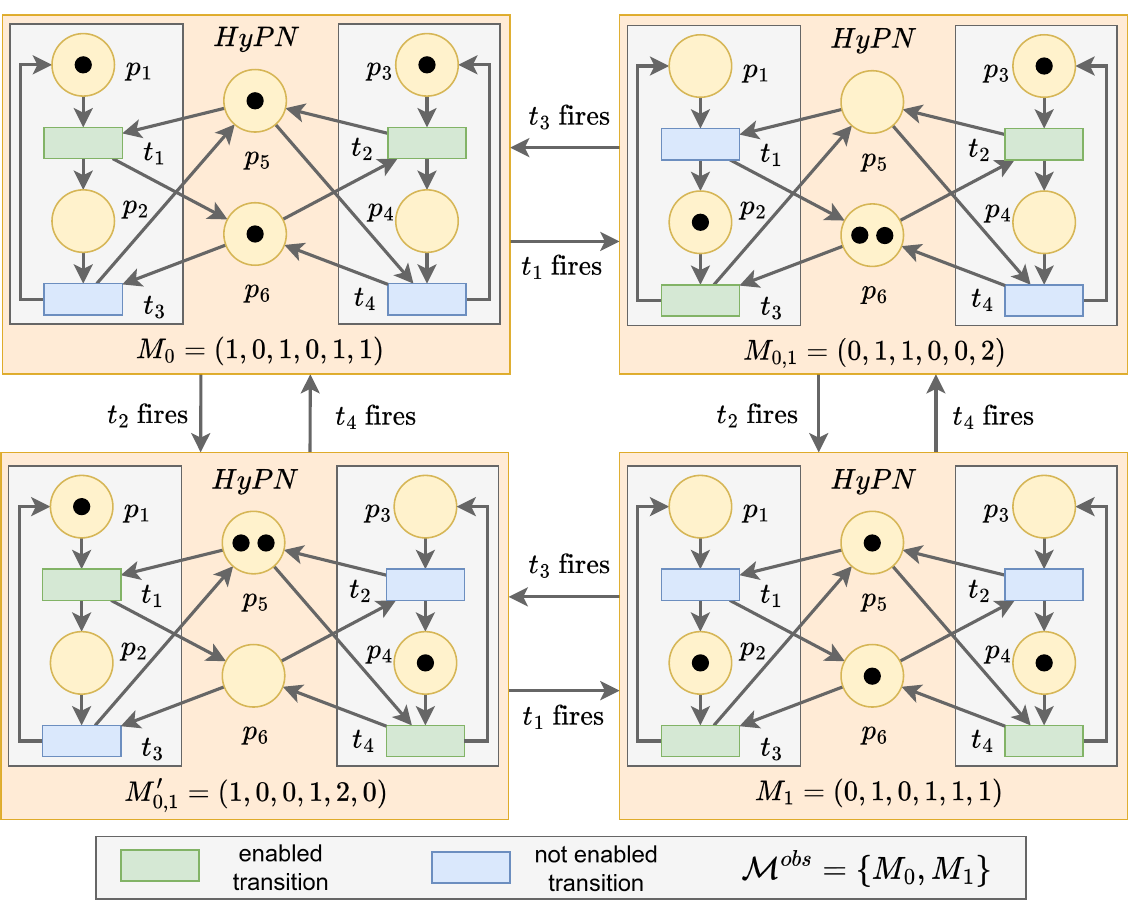}
	\caption{Example HyPN illustrating atomic macro-steps between observable markings. $M_0 \Rightarrow M_1$ admits sequences $\langle t_1,t_2\rangle$ and $\langle t_2,t_1\rangle$, and $M_1 \Rightarrow M_0$ admits $\langle t_3,t_4\rangle$ and $\langle t_4,t_3\rangle$.}
	\label{fig:hypn_example}
\end{figure}

\subsection{Structural implications and design rationale}
This subsection explains the main structural consequences of the execution-based definition of $\Sigma_H$ and clarifies the role of HyPN as an extension of classical PNs. HyPN extends classical PNs without changing their semantics. If every reachable marking is observable, i.e., $\mathcal{M}^{obs} = \mathcal{R}(M_0)$, and atomic sequences reduce to single transitions of the form $\langle t \rangle$, then the induced relation $\Rightarrow$ coincides with the standard PN firing relation, and the induced transition system corresponds to the reachability graph. Thus, HyPN generalizes classical PNs rather than replacing them.

Given the underlying PN $(P,T,F,W,M_0)$ together with the chosen set of observable markings $\mathcal{M}^{obs}$, the abstract transition system $\mathcal{S}_{HyPN} = (\mathcal{M}^{obs}, \Rightarrow)$ is derived from admissible firing sequences and depends only on the operational behavior of the net. The induced transition relation $\Rightarrow$ is determined by the execution semantics of the PN, and can be automatically computed via Algorithm~\ref{alg:sigma_h}. However, the choice of observable markings is a modeling decision. Observable and internal markings are not determined automatically by reachability but are defined by the designer. Atomicity is defined relative to the selected set $\mathcal{M}^{obs}$ and ensures that macro-steps represent meaningful changes at the observable level. HyPN is more than a simple grouping of firing sequences; it introduces a clear separation between observable and internal markings, a notion of hyper-enabled transitions at the observable level, and an induced abstract transition system $\mathcal{S}_{HyPN} = (\mathcal{M}^{obs}, \Rightarrow)$. Together, these elements create a two-layer semantic model that preserves the detailed behavior of the PN while providing a structured abstraction at the system level.

The relational structure induced by HyPN does not reduce the dynamics to a single transition between observable markings. For a given pair $(M_i, M_j) \in \mathcal{M}^{obs}$, there may exist multiple atomic transition sequences that realize the same observable state change. These alternative execution paths are explicitly represented in $\Sigma_H$, thereby exposing the internal execution structure. Consequently, HyPN is not just a transition-graph abstraction but a relational structure that captures execution-level alternatives, redundancy, multiple realization paths, and potential system-level robustness. The existence of multiple admissible sequences between observable markings reflects alternative ways of realizing the same observable behavior within the same PN dynamics. HyPN makes these alternatives explicit by representing them as elements of a relational structure between observable markings, rather than merely as distinct paths in a reachability graph. In the next section, we provide a constructive procedure for deriving $\mathcal{M}^{obs}$ from a Boolean specification, transforming HyPN into a constraint-driven modeling framework in which logical requirements are enforced at the observable level while preserving the intrinsic dynamics of the PN.

\subsection{Basic properties of HyPN}
\label{subsec:basic-properties-of-hypn}
This subsection presents key properties of HyPN, including its relation to classical PNs, structural isomorphism, and behavioral equivalence. These properties provide a basis for reasoning about HyPN models and enable comparison between different HyPN instances and their observable behaviors.

\begin{theorem}[Conservative extension]
\label{theorem:conservative-extension}
Let $H = (P,T,F,W,M_0,\mathcal{M}^{obs})$ be a HyPN. If $\mathcal{M}^{obs} = \mathcal{R}(M_0)$ then, the induced observable transition relation $\Rightarrow$ coincides with the standard PN firing relation.
\end{theorem}

\begin{proof}
If $\mathcal{M}^{obs} = \mathcal{R}(M_0)$, then every reachable marking is observable and no marking is treated as internal. Hence, any atomic sequence must consist of a single transition, since any longer sequence would pass through intermediate observable markings and violate atomicity. Therefore, the induced relation $\Rightarrow$ coincides with the standard PN firing relation.
\end{proof}

\begin{theorem}[Observable soundness]
\label{theorem:observable-soundness}
Let $M_i, M_j \in \mathcal{M}^{obs}$. If $M_i \Rightarrow M_j$, then $M_j \in \mathcal{R}(M_i)$ in the underlying PN.
\end{theorem}

\begin{proof}
If $M_i \Rightarrow M_j$, then by definition of $\Rightarrow$ there exists a sequence $\sigma \in T^*$ such that $atomic(M_i,\sigma)$ and execution of $\sigma$ from $M_i$ yields $M_j$. By definition of atomicity, $\sigma$ is admissible and therefore realizable. Hence, there exists a marking sequence $M_i = M_{i,0} \xrightarrow{t^{(1)}} \cdots \xrightarrow{t^{(k)}} M_{i,k} = M_j$, where each transition is locally enabled. Therefore, $M_j \in \mathcal{R}(M_i)$.
\end{proof}

\begin{definition}[Structural isomorphism]
\label{definition:structural-isomorphism}
Let $HyPN_i = (P_i,T_i,F_i,W_i,M_{0,i},\mathcal{M}^{obs}_i)$ and 
$HyPN_j = (P_j,T_j,F_j,W_j,M_{0,j},\mathcal{M}^{obs}_j)$ be two HyPNs. They are structurally isomorphic, denoted $HyPN_i \cong HyPN_j$, if there exist bijections $\phi_P : P_i \to P_j$ and $\phi_T : T_i \to T_j$ such that $(x,y) \in F_i \Longleftrightarrow (\phi(x),\phi(y)) \in F_j$ for all $(x,y) \in F_i$, where:
\[
\phi(z) =
\begin{cases}
	\phi_P(z), & z \in P_i,\\
	\phi_T(z), & z \in T_i,
\end{cases}
\]
and $W_i(x,y) = W_j(\phi(x),\phi(y))$ for all $(x,y) \in F_i$. Thus, structural isomorphism preserves the place-transition structure and arc weights up to renaming of elements.
\end{definition}

\begin{definition}[Behavioral equivalence]
\label{definition:behavioral-equivalence}
The behavior of a HyPN is given by the observable transition system $S_{HyPN} = (M^{obs}, \Rightarrow, M_0)$.
Let $H_i = (P_i, T_i, F_i, W_i, M_{0,i}, \mathcal{M}^{obs}_i)$ and $H_j = (P_j, T_j, F_j, W_j, M_{0,j}, \mathcal{M}^{obs}_j)$ be two HyPNs. The nets $H_i$ and $H_j$ are said to be behaviorally equivalent, denoted by $H_i \approx H_j$, if the following conditions hold:
\begin{enumerate}
	\item \textbf{Structural isomorphism:} The underlying PN structures are structurally isomorphic, i.e., there exist bijections between places and transitions that preserve arcs and arc weights.
	\item \textbf{Reachability of projected initial markings:} Let $M_{0,i}^{j}$ denote the projection of the initial marking $M_{0,i}$ of $H_i$ onto $H_j$ under the structural isomorphism, and let $M_{0,j}^{i}$ denote the projection of $M_{0,j}$ onto $H_i$.
	Then $M_{0,j} \in R(M_{0,i}^{j})$ and $M_{0,i} \in R(M_{0,j}^{i})$, where $R(\cdot)$ denotes the reachability set of the underlying PN.
	\item \textbf{Equality of observable behaviors:} The observable transition systems $\mathcal{S}_{HyPN,i}$ and $\mathcal{S}_{HyPN,j}$ coincide under the structural isomorphism.
\end{enumerate}
Consequently, the two nets induce identical observable marking sets and equivalent observable transition relations.
\end{definition}			

\section{Model synthesis}
\label{sec:model-synthesis}

This section presents a constraint-driven synthesis of a HyPN from a propositional specification $\varphi$ of invariant properties. Constraints are encoded into the net such that all observable markings satisfy $\varphi$, while intermediate markings follow execution semantics. The construction is clause-wise and compositional, with $H$ denoting the final net and $H_i$ clause-specific subnets. The overall synthesis concept is illustrated in Fig.~\ref{fig:model_synthesis_motivation}.

\begin{figure}
	\centering
	\includegraphics[width=0.9\linewidth]{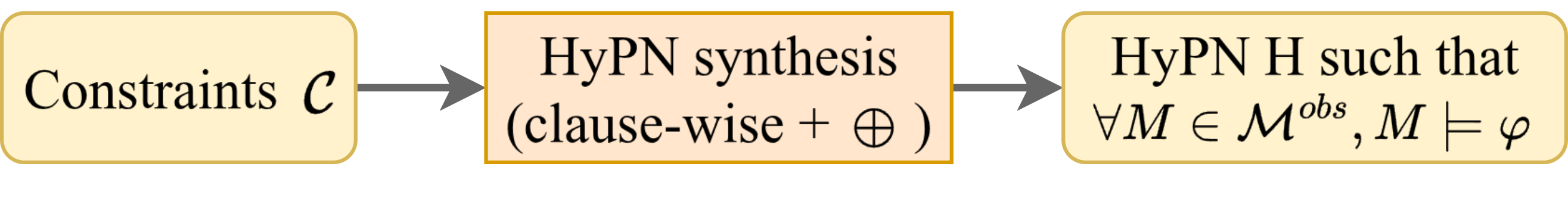}
	\caption{HyPN synthesis from constraints $\mathcal{C}$ yielding $H \models \varphi$.}
	\label{fig:model_synthesis_motivation}
\end{figure}

\subsection{Synthesis algorithm}
\label{subsec:synthesis-algorithm}

Let $V$ denote the set of propositional variables. Each variable $x \in V$ induces two complementary literals: a positive literal $x$ and a negative literal $\bar{x}$. We denote the sets of positive and negative literals by $L := \{x \mid x \in V\}$ and $\bar{L} := \{\bar{x} \mid x \in V\}$, respectively, and define the set of all literals as $\mathcal{L} = L \cup \bar{L}$. Let $\mathcal I_V = [V \rightarrow \mathbb B]$ be the set of all possible interpretations, meaning assignments of each variable to a truth value. For a propositional formula $\varphi(V)$ over $V$ and an interpretation $\nu \in \mathcal I_V$, we write $\nu \models \varphi(V)$ if the formula holds when each variable is replaced by its assigned value according to $\nu$. 

\begin{problem}[Synthesis problem]
Given a propositional formula $\varphi$ over the set of propositional variables $V$, construct a HyPN $H$ and a mapping $m : \mathcal{M}^{\mathrm{obs}} \rightarrow \mathcal I_V$ such that $\forall M \in \mathcal{M}^{\mathrm{obs}}$, $m(M) \models \varphi$, also written as $H \models \varphi$.   
\end{problem}

We start from a set of propositional constraints $\mathcal{C}$ over variables $V$, and we define $\varphi = \bigwedge_{C_i \in \mathcal{C}} C_i$. The formula $\varphi$ is then transformed into conjunctive normal form $\varphi^{\mathrm{CNF}} = \bigwedge_{i=1}^{m} \varphi_i$, where each clause has the form $\varphi_i = \bigvee_{l \in S_i} l$, and $S_i \subseteq \mathcal{L}$ denotes the set of literals in $\varphi_i$ (e.g., $\varphi_i = a \lor b \lor \bar{c}$ gives $S_i = \{a, b, \bar{c}\}$). The synthesis constructs a clause-specific HyPN for each clause $\varphi_i$ (Sec.~\ref{subsec:clause-subnet-synthesis}) and composes them into a single network using the composition operator $\oplus$ (Sec.~\ref{subsec:pn-composition}), as summarized in Algorithm~\ref{alg:hypn_synthesis}.

\begin{algorithm}
	\caption{HyPN synthesis from $\varphi^{\mathrm{CNF}}$}
	\label{alg:hypn_synthesis}
	\textbf{Input:} $\varphi^{\mathrm{CNF}} = \bigwedge_{i=1}^{m} \varphi_i$ \\
	\textbf{Output:} HyPN $H$
	\begin{algorithmic}[1]
		
		\For{each clause $\varphi_i$}
		\State Construct $H_i$ such that $H_i \models \varphi_i$ \Comment{Sec.~\ref{subsec:clause-subnet-synthesis}}
		\EndFor
		
		\State $H \gets \bigoplus_{i=1}^{m} H_i$ \Comment{Sec.~\ref{subsec:pn-composition}}
		
		\State \Return $H$
	\end{algorithmic}
\end{algorithm}
Algorithm~\ref{alg:hypn_synthesis} constructs clause-specific HyPNs $H_i$ for each $\varphi_i$ and composes them into $H$ using $\oplus$, yielding $H \models \varphi$.

\subsection{Clause-level HyPN synthesis}
\label{subsec:clause-subnet-synthesis}
This subsection defines a clause-specific HyPN $H_i$ for a clause $\varphi_i=\bigvee_{l_j\in S_i} l_j$, where $S_i \subseteq \mathcal{L}$ is a set of literals in $\varphi_i$, and $V_i := \{ x \mid x \text{ or } \bar{x} \in S_i \}$ the corresponding variables. The subnet is defined as $H_i=(P_i,T_i,F_i,W_i,M_{0,i},\mathcal{M}_i^{\mathrm{obs}})$ with:
\[
P_i = \{p_x, p_{\bar{x}} \mid x \in V_i\} \cup \{p_i^{\mathrm{sat}}\},
\quad
T_i = \{t_{x,\bar{x}}, t_{\bar{x},x} \mid x \in V_i\}.
\]
Here, $\{p_x, p_{\bar{x}} \mid x \in V_i\}$ defines the set of literal places $P_i^{\mathrm{lit}}$ and $p_i^{\mathrm{sat}}$ is a guard place for $\varphi_i$.
The arc set $F_i$ is defined as:
\[
F_i = \bigcup_{x \in V_i} \{(p_x, t_{x,\bar{x}}), (t_{x,\bar{x}}, p_{\bar{x}}),
(p_{\bar{x}}, t_{\bar{x},x}), (t_{\bar{x},x}, p_x)\}
\]
\[
\cup \bigcup_{l_j \in S_i} \{(p_i^{\mathrm{sat}}, t_{l_j,\bar{l}_j}), (t_{\bar{l}_j,l_j}, p_i^{\mathrm{sat}})\}.
\]
All arcs have unit weight. The first term connects, for each variable $x \in V_i$, complementary literal places via switching transitions, enforcing the invariant $M(p_x) + M(p_{\bar{x}}) = 1$ for all $x \in V_i$ and all reachable markings. The second term connects the guard place $p_i^{\mathrm{sat}}$ with transitions associated with literals in $S_i$, enforcing the clause-level invariant $\sum_{l_j \in S_i} M(p_{\bar{l}_j}) + M(p_i^{\mathrm{sat}}) = |S_i|$.
A marking satisfies $\varphi_i$ iff $M(p_i^{\mathrm{sat}}) > 0$.
The initial marking $M_{0,i}$ is defined such that for each literal $l_j \in S_i$, $M_{0,i}(p_{l_j}) = 1$ and $M_{0,i}(p_{\bar{l}_j}) = 0$. Additionally, $M_{0,i}(p_i^{\mathrm{sat}}) = |S_i|$. For example, for $\varphi_i = l_A \lor \bar{l}_B$, the initial marking satisfies $M_{0,i}(p_{l_A})=1$ and $M_{0,i}(p_{\bar{l}_B})=1$, while all complementary places are unmarked, as in Fig.~\ref{fig:pn_patterns_based_on_place_invariants_negative}(b).

The set of observable markings is defined as:
\[
\mathcal{M}_i^{\mathrm{obs}} = \{ M \in \mathcal{R}(M_{0,i}) \mid M(p_i^{\mathrm{sat}}) > 0 \}.
\]
For a marking $M \in \mathcal{M}_i^{\mathrm{obs}}$ and a variable $x \in V_i$, $m(M)(x)$ is true if there is a token in $p_x$, and false otherwise. By construction, every observable marking satisfies $\varphi_i$, i.e., $\forall M \in \mathcal{M}_i^{\mathrm{obs}},\; M \models \varphi_i$, and thus $H_i \models \varphi_i$. Thus, $H_i$ encodes the satisfaction condition of $\varphi_i$ and serves as a building block for the final HyPN $H$. Examples of clause-specific subnets are shown in Fig.~\ref{fig:pn_patterns_based_on_place_invariants_negative}: (a) a clause composed of positive literals (e.g., $l_A \lor l_B$), and (b) a clause containing both positive and negative literals (e.g., $l_A \lor \bar{l}_B$).

\begin{figure}[t]
	\centering
	\includegraphics[width=0.85\linewidth]{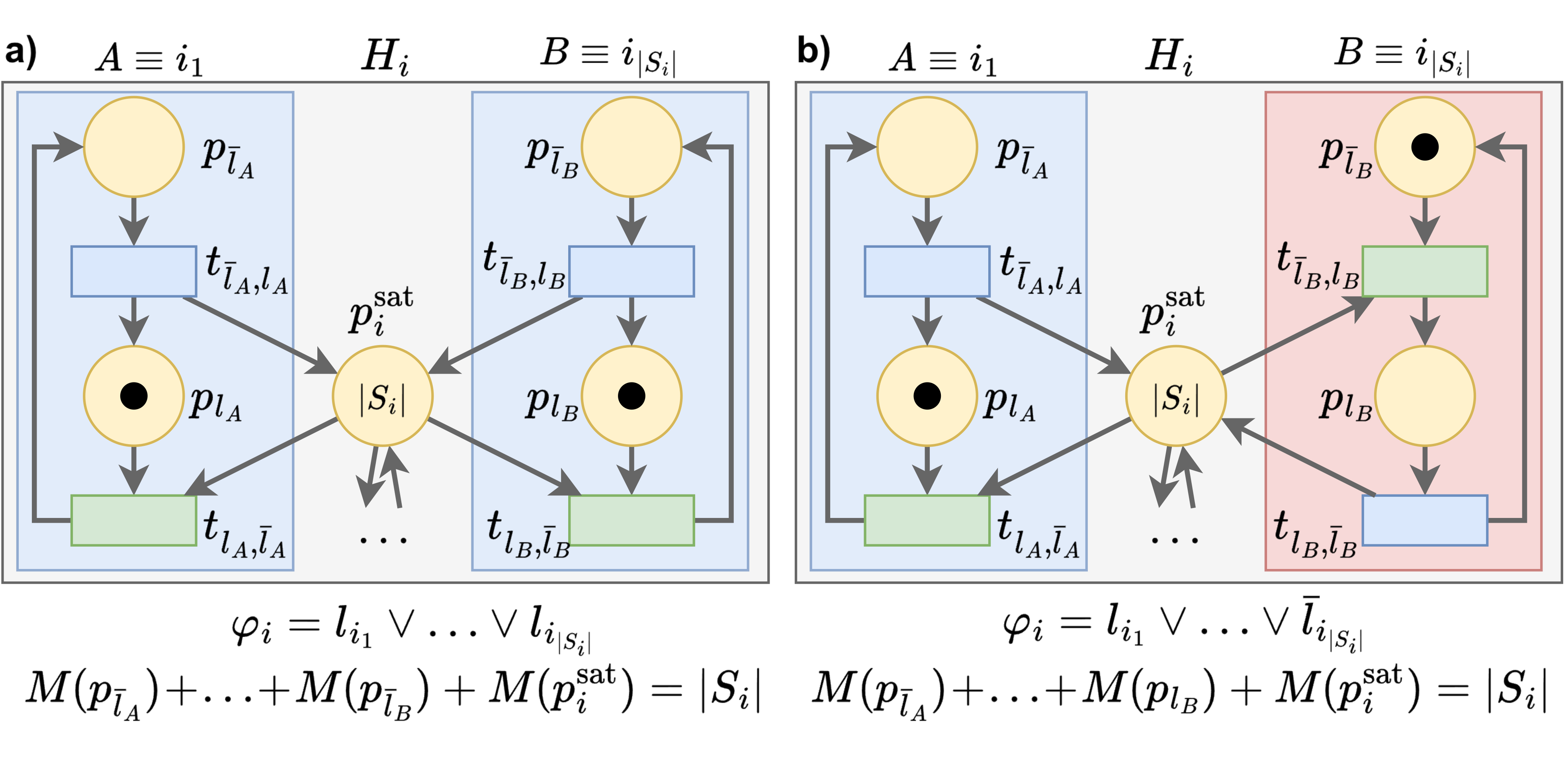}
	\caption{Examples of clause-specific HyPN subnets: (a) clause composed of positive literals; (b) clause containing both positive and negative literals. Each subnet enforces clause satisfaction at observable markings.}	
	\label{fig:pn_patterns_based_on_place_invariants_negative}
\end{figure}

\subsection{HyPN composition}
\label{subsec:pn-composition}
Given subnets $\{H_i\}$ encoding $\varphi_i$, the goal is to construct a HyPN $H$ with observable markings satisfying $\varphi = \bigwedge_i \varphi_i$. The composition operator $\oplus$ combines $H_i$ and $H_j$, preserving $(\varphi_i \land \varphi_j)$, with $H_{\{i,j\}} = H_i \oplus H_j$. The final HyPN is $H = \bigoplus_i H_i$ (see Fig.~\ref{fig:general_idea_of_pn_composition}). The operator $\oplus$ satisfies closure, commutativity, and associativity (see Sec.~\ref{sec:formal-properties}), ensuring well-definedness.

\begin{figure}
	\centering
	\begin{subfigure}{0.49\linewidth}
		\includegraphics[width=1.0\linewidth]{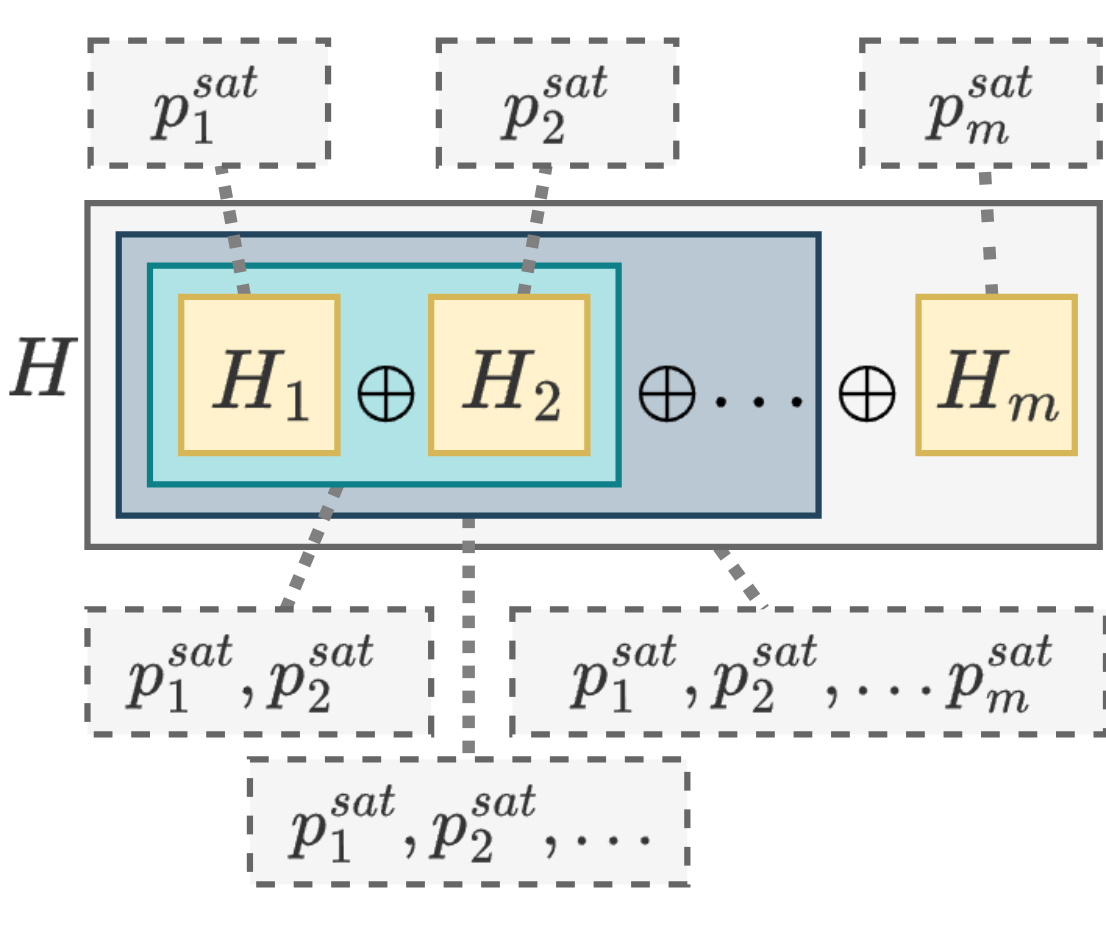}
		\caption{}
		\label{fig:general_idea_of_pn_composition}
	\end{subfigure}
	\begin{subfigure}{0.49\linewidth}
		\includegraphics[width=1.0\linewidth]{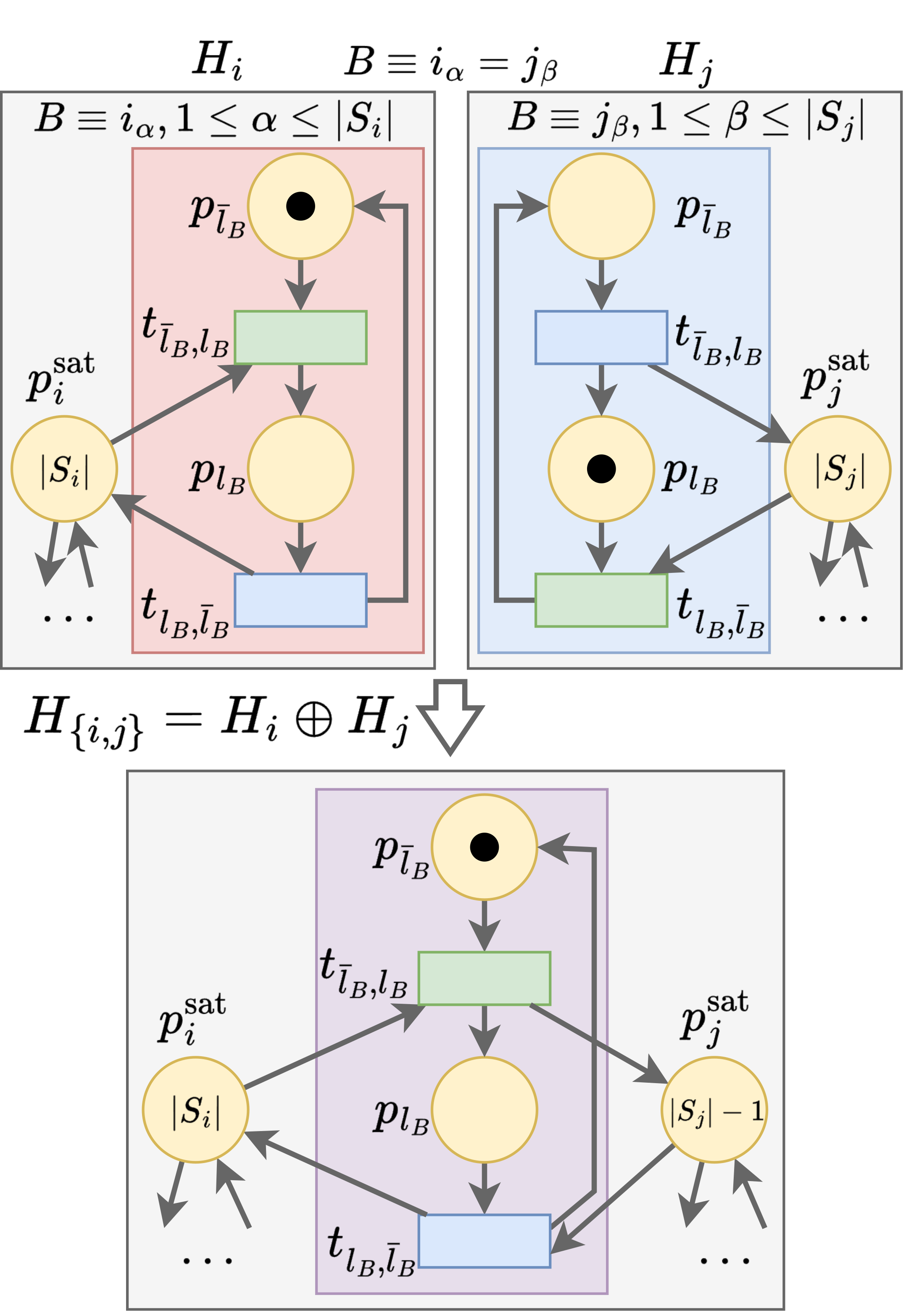}
		\caption{}
		\label{fig:pn_patterns_combine_two_pn_mixed}
	\end{subfigure}
	\caption{(a) Clause-wise HyPN composition with $\oplus$; (b) binary composition $H_i \oplus H_j$ with conflicting literals, $M_i(p_{\bar{l}_B}) = 1 \land M_j(p_{l_B}) = 1$.}
	\label{fig:general_idea_of_pn_composition_and_pn_patterns_combine_two_pn_mixed}
\end{figure}

\textit{Binary composition of HyPNs.} 
Let $H_i=(P_i,T_i,F_i,W_i,M_{0,i},\mathcal{M}^{\mathrm{obs}}_i)$ and $H_j=(P_j,T_j,F_j,W_j,M_{0,j},\mathcal{M}^{\mathrm{obs}}_j)$ be two HyPNs. Their composition is $H_{\{i,j\}} = H_i \oplus H_j$, where:
\begin{equation*}
	\begin{aligned}
		H_{\{i,j\}} =
		(
		& P_i \cup P_j,\;
		T_i \cup T_j,\;
		F_i \cup F_j,\\
		& W_i \cup W_j,\;
		M_{0,i} \oplus M_{0,j},\;
		\mathcal{M}^{\mathrm{obs}}_{\{i,j\}})
	\end{aligned}
\end{equation*}
and the initial marking is $M_{0,\{i,j\}} := M_{0,i} \oplus M_{0,j}$, where $\oplus$ aligns tokens while preserving literal consistency. For a composed HyPN $H_I$, let $P_I = P^{lit}_I \cup P^{sat}_I$, where $P^{sat}_I = \{p^{sat}_i \mid i \in I\}$ and $P^{lit}_I \cap P^{sat}_I = \varnothing$. For clause-level nets, $P^{sat}_i \cap P^{sat}_j = \varnothing$ for $i \neq j$. Hence, overlaps between $P_i$ and $P_j$ may occur only in literal places. The construction of $M_{0,\{i,j\}}$ reduces to aligning initial markings on shared literal places, with modification of guard-place markings. This alignment is defined by the following cases.

\noindent
\textit{Case 1 (No shared literal places)} If $P^{lit}_i \cap P^{lit}_j = \emptyset$, the composed marking is given by direct union:
\[
(M_{0,i} \oplus M_{0,j})(p)
=
\begin{cases}
	M_{0,i}(p), & p \in P_i, \\
	M_{0,j}(p), & p \in P_j.
\end{cases}
\]

\noindent
\textit{Case 2 (Shared literals without conflict).} 
If $P^{lit}_i \cap P^{lit}_j \neq \emptyset$ and $\forall p \in P^{lit}_i \cap P^{lit}_j:\; M_{0,i}(p) = M_{0,j}(p)$, the composed marking is as in Case~1.

\noindent
\textit{Case 3 (Literal conflict and swap).} 
If there exists $p \in P^{lit}_i \cap P^{lit}_j$ such that $M_{0,i}(p) \neq M_{0,j}(p)$, define the conflict set $P^{conf}_{ij} = \{\, p \in P^{lit}_i \cap P^{lit}_j \mid M_{0,i}(p) \neq M_{0,j}(p) \}$. Fig.~\ref{fig:pn_patterns_combine_two_pn_mixed} shows a representative conflict between complementary literals. Each $p_l \in P^{conf}_{ij}$ is resolved by selecting a swap transition:
\[
t^*_l =
\begin{cases}
	t_{\bar{l},l}, & \text{if } M_{0,i}(p_l)=1,\\
	t_{l,\bar{l}}, & \text{otherwise},
\end{cases}
\]
aligning the marking of $H_j$ with $H_i$. Firing $t^*_l$ updates literal and guard-place markings and is admissible if $M(p^{sat}) \ge 1$. If admissible, both $p_l$ and $p_{\bar l}$ are removed from $P^{conf}_{ij}$; otherwise, a symmetric swap is attempted in the other net. If neither swap is admissible, the clauses are incompatible. Swaps continue until $P^{conf}_{ij} = \emptyset$, reducing the problem to Case~2 with respect to literal places. The composed marking is then given by:
\[
(M_{0,i} \oplus M_{0,j})(p)=
\begin{cases}
	M_i(p), & p \in P_i,\\
	M_j(p), & p \in P_j \setminus P_i.
\end{cases}
\]

Observable markings of binary composition are defined as:
\begin{equation*}
	\mathcal{M}^{\mathrm{obs}}_{\{i,j\}}
	= 
	\left\{
	M \in \mathcal{R}(M_{0,\{i,j\}})
	\;\middle|\;
	\forall p \in P^{sat}_{\{i,j\}} : M(p) \ge 1
	\right\},
\end{equation*}
where $P^{sat}_{\{i,j\}} := P^{sat}_i \cup P^{sat}_j$. Observable markings satisfy the composed specification.


\textit{Composition of a finite set of HyPNs.} 
It can be shown that the $\oplus$ operator is commutative and associative (see Section \ref{subsec:formal-properties-algebraic-properties}). Thus, given $\varphi^{CNF} = \bigwedge_{i=1}^{k} \varphi_i$, define $H_S = (P_S, T_S, F_S, W_S, M_{0,S}, \mathcal{M}^{\mathrm{obs}}_S) := \bigoplus_{i=1}^{k} H_i$ where each $H_i$ is synthesized from $\varphi_i$ using the procedure described in Section \ref{subsec:clause-subnet-synthesis}. Let $P^{sat}_S = \{p^{sat}_i \mid i=1..k\}$. The observable markings are defined as:
\[
\mathcal{M}^{\mathrm{obs}}_{S}
=
\left\{
M \in \mathcal{R}(M_{0,S})
\;\middle|\;
\forall p \in P^{sat}_{S} : M(p) \ge 1
\right\}.
\]
Observable markings satisfy $\varphi^{CNF}$, i.e., $H_S \models \varphi^{CNF}$.

\subsection{Example}
\label{subsec:modeling_synthesis_example}
Consider the constraint $\varphi = (a \Leftrightarrow b)$ with CNF clauses $\varphi_1 = \bar a \lor b$ and $\varphi_2 = a \lor \bar b$, where $S_1 = \{\bar a, b\}$, $S_2 = \{a, \bar b\}$, and $V_1 = V_2 = \{a, b\}$. Clause-level synthesis results in HyPNs $H_1$ and $H_2$, whose composition $H_{\{1,2\}} = H_1 \oplus H_2$ is shown in Fig.~\ref{fig:hypn_example_modeling_synthesis}, with initial marking $M_{0,\{1,2\}} = (1,0,1,0,1,1)$. Two swap operations resolve conflicting literal markings. The observable markings are $\mathcal{M}^{\mathrm{obs}}_{\{1,2\}} = \{(1,0,1,0,1,1), (0,1,0,1,1,1)\}$, which coincide with the satisfying valuations of $\varphi$ (see Fig.~\ref{fig:hypn_example}), with all clause-satisfaction places marked, while other reachable markings are non-observable. From $M_{0,\{1,2\}}$, any single transition leads to a non-observable marking, and observable markings are reached only after sequences of firings satisfying all clauses.

\begin{figure}
	\centering
	\includegraphics[width=1.0\linewidth]{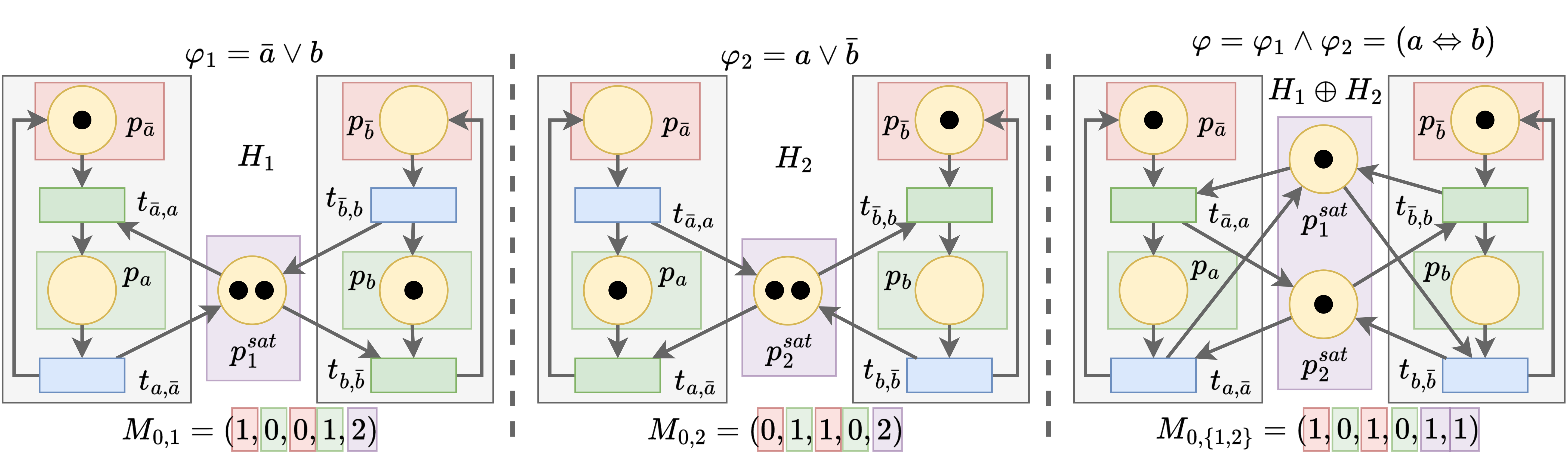}
	\caption{Clause-level HyPNs synthesized for $\varphi_1 = \bar a \lor b$ and $\varphi_2 = a \lor \bar b$, and their composition $H_1 \oplus H_2$, yielding a HyPN encoding $\varphi = a \Leftrightarrow b$, i.e., $H_{\{1,2\}} \models (a \Leftrightarrow b)$.}
	\label{fig:hypn_example_modeling_synthesis}
\end{figure}

\subsection{Discussion on structural limitations}
\label{subsec:synthesis-discussion}

While the synthesis yields a compact PN structure, it intentionally does not enforce constraint satisfaction at all reachable markings, thereby avoiding the combinatorial blow-up associated with explicit modeling of literal-level dynamics. For the disjunctive constraint $\varphi = a_1 \lor a_2 \lor \dots \lor a_n$, where $n = |V|$, all Boolean valuations except the all-false one are admissible, resulting in $2^n - 1$ states. Enforcing direct transitions between any two admissible states induces a complete directed graph, yielding $(2^n - 1)(2^n - 2)$ transitions. This exponential growth renders modeling of literal dynamics impractical even for moderate $n$. In contrast, the proposed synthesis yields a linear-size PN with $2n$ transitions by abstracting away explicit inter-literal transitions and delegating their resolution to execution semantics. As a result, intermediate markings may temporarily violate the constraint, while correctness is preserved at observable markings. Such intermediate violations are essential for scalability. Rather than enforcing logical satisfaction at every reachable marking, correctness is ensured at the execution level: clause-level guard places combined with macro-step execution semantics permit temporary violations while guaranteeing that all observable markings satisfy $\varphi$. 		

\section{Formal properties and scalability analysis}
\label{sec:formal-properties}

We restrict the analysis to a subclass of HyPNs corresponding to the encoding of logical subformulas introduced in Section~\ref{sec:model-synthesis}. Let $\mathcal{H}$ denote the class of such HyPNs. Each clause-generated subnet, denoted $H_i$, belongs to $\mathcal{H}$.

This section analyzes the formal properties of the HyPN composition operator $\oplus$ introduced in Section~\ref{subsec:pn-composition}. We establish its algebraic properties, ensuring that HyPNs synthesized from individual constraints can be composed modularly and in arbitrary order. We then examine the semantic relationship between the synthesized HyPN and the logical constraint formula it encodes, showing that observable markings correspond exactly to satisfying valuations. Finally, we briefly discuss scalability and structural complexity of the resulting networks.

\subsection{Algebraic properties of the composition operator}
\label{subsec:formal-properties-algebraic-properties}

\begin{lemma}[Well-defined initial marking composition]
\label{lemma:marking-composition}
Let $H_i, H_j \in \mathcal{H}$. Then the composed initial marking $M_{0,\{i,j\}} = M_{0,i} \oplus M_{0,j}$ is a well-defined marking over $P_i \cup P_j$.
\end{lemma}

\begin{proof}
It suffices to show that $(M_{0,i} \oplus M_{0,j})(p)$  is uniquely defined for every $p \in P_i \cup P_j$. If $p \in P_i \setminus P_j$ (resp. $p \in P_j \setminus P_i$), the value is inherited from $M_{0,i}(p)$ (resp. $M_{0,j}(p)$), while for $p \in P_i \cap P_j$ it is given by the literal-composition rule. Conflicts are resolved using the swap mechanism using a deterministic policy: a feasible swap in $H_j$ is applied, otherwise in $H_i$, and if neither is feasible, the conflict remains unresolved. Hence, $(M_{0,i} \oplus M_{0,j})(p)$ is uniquely determined for all $p \in P_i \cup P_j$, and $M_{0,\{i,j\}}$ is a valid marking over $P_i \cup P_j$.
\end{proof}


\begin{theorem}[Closure]
	\label{theorem:closure}
Let $H_i, H_j \in \mathcal{H}$. Then $H_i \oplus H_j \in \mathcal{H}$.
\end{theorem}

\begin{proof}
Let $H_k = H_i \oplus H_j$. It suffices to show that the composition preserves the constraint-structured form. By definition of $\oplus$, $H_k = (P_i \cup P_j,\; T_i \cup T_j,\; \dots,\; M_{0,i} \oplus M_{0,j},\; \mathcal{M}^{obs}_{\{i,j\}})$. Thus, the structural components are preserved, and by Lemma~\ref{lemma:marking-composition}, the initial marking is well-defined. Since $H_i, H_j \in \mathcal{H}$, $H_i \models \varphi_i$ and $H_j \models \varphi_j$. By definition, $\mathcal{M}^{obs}_{\{i,j\}} =  \{ M \in R(M_{0,\{i,j\}}) \mid \forall p \in (P^{sat}_i \cup P^{sat}_j): M(p) \ge 1 \}$. Thus, every observable marking satisfies both clause conditions, hence $\varphi_i \land \varphi_j$. Therefore, $H_k \models (\varphi_i \land \varphi_j)$, and $H_k \in \mathcal{H}$.
\end{proof}
	
\begin{theorem}[Commutativity] 
	\label{theorem:commutativity}
	Let $H_i, H_j \in \mathcal{H}$. Then $H_i \oplus H_j$ and $H_j \oplus H_i$ are behaviorally equivalent ($H_i \oplus H_j \approx H_j \oplus H_i$).
\end{theorem}

\begin{proof}
	The composition operator $\oplus$ is defined component-wise. Since set union is commutative, the structural components satisfy $P_i \cup P_j = P_j \cup P_i,\quad T_i \cup T_j = T_j \cup T_i,\quad F_i \cup F_j = F_j \cup F_i,\quad W_i \cup W_j = W_j \cup W_i$, and hence the composed nets have identical structure. 
	The composed initial markings $M_{0,i} \oplus M_{0,j}$ and $M_{0,j} \oplus M_{0,i}$ may differ due to conflict resolution, including in the markings of satisfaction places. However, for each literal conflict, the composed net contains complementary swap transitions $t_{l,\hat{l}}$ and $t_{\hat{l},l}$ (cf. $t^*_l$ in Case 3 of HyPN composition) induced by $T_i \cup T_j$. Hence, differences between the composed initial markings can be resolved by firing appropriate complementary swap transitions.
	
	By definition of observable markings,
	$\mathcal{M}^{obs}_{\{i,j\}} = \{ M \in R(M_{0,\{i,j\}}) \mid \forall p \in (P^{sat}_i \cup P^{sat}_j): M(p) \ge 1 \}$,
	which is invariant under permutation of $i$ and $j$. Since the two compositions have identical structure and differ due to conflict resolution handled by swap transitions, they yield identical sets of observable markings. Hence, $H_i \oplus H_j \approx H_j \oplus H_i$.
\end{proof}

\begin{theorem}[Associativity]
	\label{theorem:associativity}
	Let $H_i, H_j, H_k \in \mathcal{H}$. Then $(H_i \oplus H_j) \oplus H_k \approx H_i \oplus (H_j \oplus H_k)$.
\end{theorem}

\begin{proof}
	The structural components are obtained by set unions: $(P_i \cup P_j) \cup P_k = P_i \cup (P_j \cup P_k)$, and	$(T_i \cup T_j) \cup T_k = T_i \cup (T_j \cup T_k)$,	and analogously for arcs and weights. Hence, both $(H_i \oplus H_j) \oplus H_k$ and $H_i \oplus (H_j \oplus H_k)$ yield the same net structure and satisfaction places $P_i^{sat} \cup P_j^{sat} \cup P_k^{sat}$. Let $M_{0,L}$ and $M_{0,R}$ be the initial markings of $(H_i \oplus H_j) \oplus H_k$ and $H_i \oplus (H_j \oplus H_k)$, respectively. They may differ due to different conflict-resolution orders on shared literal places, but both satisfy $M(p_l)+M(p_{\bar l})=1$ for each $(p_l,p_{\bar l})$.
	
	The difference between $M_{0,L}$ and $M_{0,R}$ is a finite set of literal pairs with different markings. For each pair, the marking can be changed by firing a complementary swap transition $t_{l,\bar l}$ or $t_{\bar l,l}$. They preserve the invariant and modify a single pair $(p_l,p_{\bar l})$. Although they affect satisfaction places, they correct the differences between literal pairs. Hence, there is a finite firing sequence from $M_{0,L}$ to $M_{0,R}$, and by symmetry from $M_{0,R}$ to $M_{0,L}$, implying $\mathcal{R}(M_{0,L})=\mathcal{R}(M_{0,R})$.
		
	$\mathcal{M}^{obs} = \{\, M \in \mathcal{R}(M_0) \mid \forall p \in (P_i^{sat}\cup P_j^{sat}\cup P_k^{sat}) : M(p)\geq 1 \,\}$.	Since both share the same structure and reachable markings, they induce the same set of observable markings. Hence, $(H_i \oplus H_j)\oplus H_k \approx H_i \oplus (H_j \oplus H_k)$.
\end{proof}

\begin{theorem}[Order independence of composition]
	\label{theorem:permutation}
	Let $H_1,\ldots,H_n \in \mathcal{H}$. For permutation $\pi$ of $\{1,\ldots,n\}$, the composed net: $H_{\pi(1)} \oplus H_{\pi(2)} \oplus \cdots \oplus H_{\pi(n)} \approx H_1 \oplus H_2 \oplus \cdots \oplus H_n$.
\end{theorem}

\begin{proof}
	The result follows directly from Theorems~\ref{theorem:commutativity} and~\ref{theorem:associativity}, since permutation can be obtained by swapping neighbors, with associativity ensuring grouping is irrelevant.
\end{proof}

\subsection{Semantic properties of the synthesized HyPN}
\label{subsec:formal-properties-semantic-properties}

In this section we will use the following notation. For a marking $M \in \mathcal{M}^{\mathrm{obs}}$, we denote $m(M)$ by $v_M$. Also, we write $Val(\varphi) = \{v \in \mathcal I_V | v \models \varphi\}$.

\begin{theorem}[Soundness of HyPN encoding]
	\label{theorem:soundness}	
	Let $\varphi$ be the constraint formula encoded by a synthesized HyPN $H$. Then for every observable marking $M \in \mathcal{M}^{obs}$, the induced valuation $v_M$ satisfies $\varphi$, i.e., $M \in \mathcal{M}^{obs} \;\Rightarrow\; v_M \models \varphi$.
\end{theorem}

\begin{proof}
	Let $\varphi = \bigwedge_{i \in S} \varphi_i$ be the CNF representation of the encoded formula, where each clause $\varphi_i$ is represented by a clause-specific subnet $H_i$ with clause-satisfaction place $p_i^{sat}$. Consider any observable marking $M \in \mathcal{M}^{obs}$, and let $v_M$ be the valuation induced by $M$. By definition of observable markings, $M(p_i^{sat}) \ge 1$ for all $i \in S$. By construction of each subnet $H_i$, the condition $M(p_i^{sat}) \ge 1$ implies that at least one literal in the clause $\varphi_i$ is true under $v_M$. Therefore, $v_M \models \varphi_i$ for all $i \in S$. Since $\varphi = \bigwedge_{i \in S} \varphi_i$, it follows that $v_M \models \varphi$.
\end{proof}

\begin{theorem}[Completeness of HyPN encoding]
	\label{theorem:completeness}
	Let $\varphi$ be the formula encoded by a synthesized HyPN $H$. Then every valuation satisfying $\varphi$ induces an observable marking, i.e., $v \models \varphi \;\Rightarrow\; \exists M_v \in \mathcal{M}^{obs}$.
\end{theorem}

\begin{proof}
	Let $\varphi = \bigwedge_{i \in S} \varphi_i$ be the CNF representation of the encoded formula, where each clause $\varphi_i$ is represented by a subnet $H_i$ with clause-satisfaction place $p_i^{sat}$. Consider any valuation $v$ such that $v \models \varphi$. Then $v \models \varphi_i$ for all $i \in S$, so each clause $\varphi_i$, at least one literal is true under $v$. Let $M_v$ be the marking induced by $v$ on the literal places of $H$. By construction of each subnet $H_i$, this implies $M_v(p_i^{sat}) \ge 1$. Therefore, $M_v(p_i^{sat}) \ge 1$ for all $i \in S$. Hence, by definition of observable markings, $M_v \in \mathcal{M}^{obs}$.
\end{proof}

\begin{theorem}[Bijection of the HyPN encoding]
	\label{theorem:bijection}
	Let $\varphi$ be the constraint formula encoded by a synthesized HyPN $H$. Then there exists a bijection $\Phi : Val(\varphi) \leftrightarrow \mathcal{M}^{obs}$.
\end{theorem}

\begin{proof}
	Let $\varphi=\bigwedge_{i\in S}\varphi_i$, $\varphi_i=\bigvee_{l\in S_i} l$. Define $\Phi: Val(\varphi)\to \mathcal{M}^{obs}$ by assigning to each $v\in Val(\varphi)$ $\Phi(v)=M_v$, where $M_v(p_x)=1 \iff v(x)=1$, $M_v(p_{\bar x})=1 \iff v(x)=0$, and for each clause $\varphi_i$, $M_v(p_i^{sat})= \left|\{\, l\in S_i \mid v \models l \,\}\right|$. Here, $v \models l$ means that $l$ is true under $v$. Since $v\models\varphi$, every clause $\varphi_i$ contains at least one literal true under $v$, hence $M_v(p_i^{sat})\ge 1$ for all $i\in S$. Therefore, $M_v\in \mathcal{M}^{obs}$, so $\Phi$ is well-defined. Next, define $\Psi:\mathcal{M}^{obs}\to Val(\varphi)$ by assigning to each observable marking $M$ the valuation $\Psi(M)=v_M$ induced by the literal places, i.e., $v_M(x)=1 \iff M(p_x)=1$. By Theorem~\ref{theorem:soundness}, $v_M\models\varphi$, so $\Psi$ is well-defined. By construction, $\Psi(\Phi(v))=v$ for all $v\in Val(\varphi)$, and $\Phi(\Psi(M))=M$ for all $M\in\mathcal{M}^{obs}$. Thus, $\Phi$ and $\Psi$ are mutually inverse, and $\Phi$ is a bijection between $Val(\varphi)$ and $\mathcal{M}^{obs}$.
\end{proof}

\subsection{Scalability and structural complexity}
The algebraic properties of the composition operator, in particular closure and associativity, enable incremental construction of the synthesized HyPN. New constraints can therefore be incorporated by composing additional clause subnets without reconstructing the entire model. The structural size of the synthesized Petri net grows linearly with the size of the constraint formula: each literal introduces a constant number of places and transitions, while each clause contributes a clause-satisfaction place and a bounded number of arcs connecting the clause subnet to the literal places. Consequently, the numbers of places and transitions satisfy $|P|, |T| = O(|L|+|C|)$, where $|L|$ and $|C|$ denote the numbers of literals and clauses in the CNF formula.

The modular structure of the synthesized net supports incremental analysis and simulation, as clause-satisfaction places are local to clause subnets, while literal places are shared across clause subnets. By Theorem~\ref{theorem:bijection}, the number of observable markings is exactly the number of satisfying valuations of the constraint formula. Although it may be large for formulas with many solutions, checking observability is straightforward, as it only requires verifying that each clause-satisfaction place $p_i^{sat}$ contains at least one token.			

\section{Lunar rover case study}
\label{sec:endurance-case-study}

We consider a case study inspired by the Endurance mission~\cite{keane_endurance_nodate}. A lunar rover operates on the far side of the Moon and must visit a sequence of points of interest $(\textrm{poi}_1,\ldots,\textrm{poi}_k)$ in the South Pole--Aitken basin. Along the traverse, it must take scientific measurements every $\Delta_m$ kilometers and communicate with Earth whenever a relay satellite provides a communication window. The mission succeeds if all points of interest are visited, no measurement or communication opportunity is missed, and the rover remains operational.

As in many cyber-physical autonomous systems, the control software is modular. A supervisory controller coordinates three subsystems: navigation, communication, and instrumentation. Figure~\ref{fig:endurance_details} illustrates this architecture. The supervisor receives a mission command $visit(\textrm{poi}_1,\ldots,\textrm{poi}_k)$ and issues commands to the subsystems: $drive(\textrm{poi}_i)$ to navigation, $comm$ to communication, and $measure$ to the instrument. In return, the subsystems provide observations $o_N$, $o_C$, and $o_I$ describing their states and the outcomes of issued commands. Such modular architectures are susceptible to coordination problems such as deadlocks and livelocks, motivating formal analysis of the integrated behavior.

\begin{figure}
    \centering
    \includegraphics[width=\linewidth]{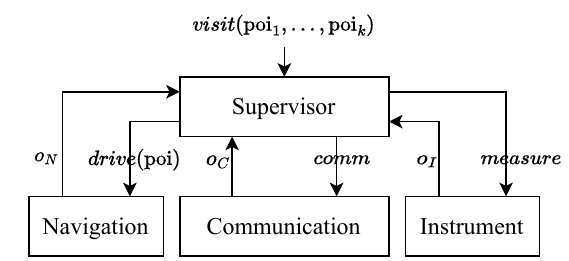}
    \caption{A high-level view of the architecture of a hypothetical lunar rover.}
    \label{fig:endurance_details}
\end{figure}

The rover is powered by a battery with state of charge $\textrm{soc}$. Since driving consumes more power than the generator can provide, the supervisor must manage energy carefully. When $\textrm{soc}$ falls below a minimum threshold $\textrm{MIN}$, the supervisor issues a $recharge$ command, placing the rover in hibernation so that generated power can restore the battery. Recharging stops once $\textrm{soc}$ reaches an upper threshold $\textrm{MAX}$. Although $\textrm{soc}$ is continuous, we use an abstraction based on three regions: below $\textrm{MIN}$, between $\textrm{MIN}$ and $\textrm{MAX}$, and above $\textrm{MAX}$.

Communication opportunities occur periodically. Because telemetry is essential for monitoring rover health and mission progress, communication is treated as a high-priority activity and should always be attempted during an available window. 

We use an abstract model in which each command is associated with a binary state variable indicating whether the corresponding activity is in progress. Specifically, $d$, $c$, $m$, and $r$ denote driving, communication, measurement, and recharging, respectively.  					

\section{Scenarios}
\label{sec:scenarios}
To illustrate the applicability of the proposed correct-by-construction synthesis algorithm (Sec.~\ref{sec:model-synthesis}), two representative scenarios are derived from the lunar rover case study presented in Sec.~\ref{sec:endurance-case-study}. The synthesized Petri net models capture the supervisory logic governing the activation, deactivation, and concurrency of rover activities, rather than their physical dynamics. For each scenario, a set of propositional constraints $\mathcal{C}$ is defined, and a HyPN $H$ is synthesized such that all observable markings $M \in \mathcal{M}^{obs}$ satisfy the global constraint formula $\varphi = \bigwedge \mathcal{C}$, that is, $M \models \varphi$. The resulting models are analyzed through their observable behavior.

\subsection{Scenario A}
\label{subsec:scenario-A}

Let the supervisor control four rover activities under supervisory control: 
1) drive ($d$); 
2) communicate ($c$), 
3) measure ($m$), and 
4) recharge ($r$). 
These activities are subject to constraints ensuring safe operation. Scenario A focuses on the fundamental mutual-exclusion constraints between activities. The constraints and synthesis steps are given below.

\subsubsection{HyPN synthesis}
\label{subsubsec:scenario-a-hypn-synthesis}

\textit{Constraints.}
Consider the set $\mathcal{C} = \{C_1, C_2, C_3, C_4\}$:
\begin{enumerate}[leftmargin=1.0cm]
	\item[$C_1$:] $d \Rightarrow (\bar{c} \land \bar{m} \land \bar{r})$: excluding all other activities,
	\item[$C_2$:] $m \Rightarrow (\bar{d} \land \bar{r})$: excluding driving and recharging,
	\item[$C_3$:] $r \Rightarrow (\bar{d} \land \bar{m} \land \bar{c})$: excluding all other activities,
	\item[$C_4$:] $c \Rightarrow (\bar{d} \land \bar{r})$: excluding driving and recharging.
\end{enumerate}

\textit{Literal sets.}
The corresponding Boolean variables are $V = \{c, d, m, r\}$, with positive literals $L = V$ and negative literals $\bar{L} = \{\bar{c}, \bar{d}, \bar{m}, \bar{r}\}$. The set of all literals is $\mathcal{L} = L \cup \bar{L}$.

\textit{Formula construction.}
The global constraint formula is defined as $\varphi = C_1 \land C_2 \land C_3 \land C_4$. Its CNF is given by:
$
\varphi^{CNF} =
(\bar{d} \lor \bar{c}) \land
(\bar{d} \lor \bar{m}) \land
(\bar{d} \lor \bar{r}) \land
(\bar{m} \lor \bar{r}) \land
(\bar{r} \lor \bar{c}).
$

\begin{figure}
	\centering
	\includegraphics[width=0.9\linewidth]{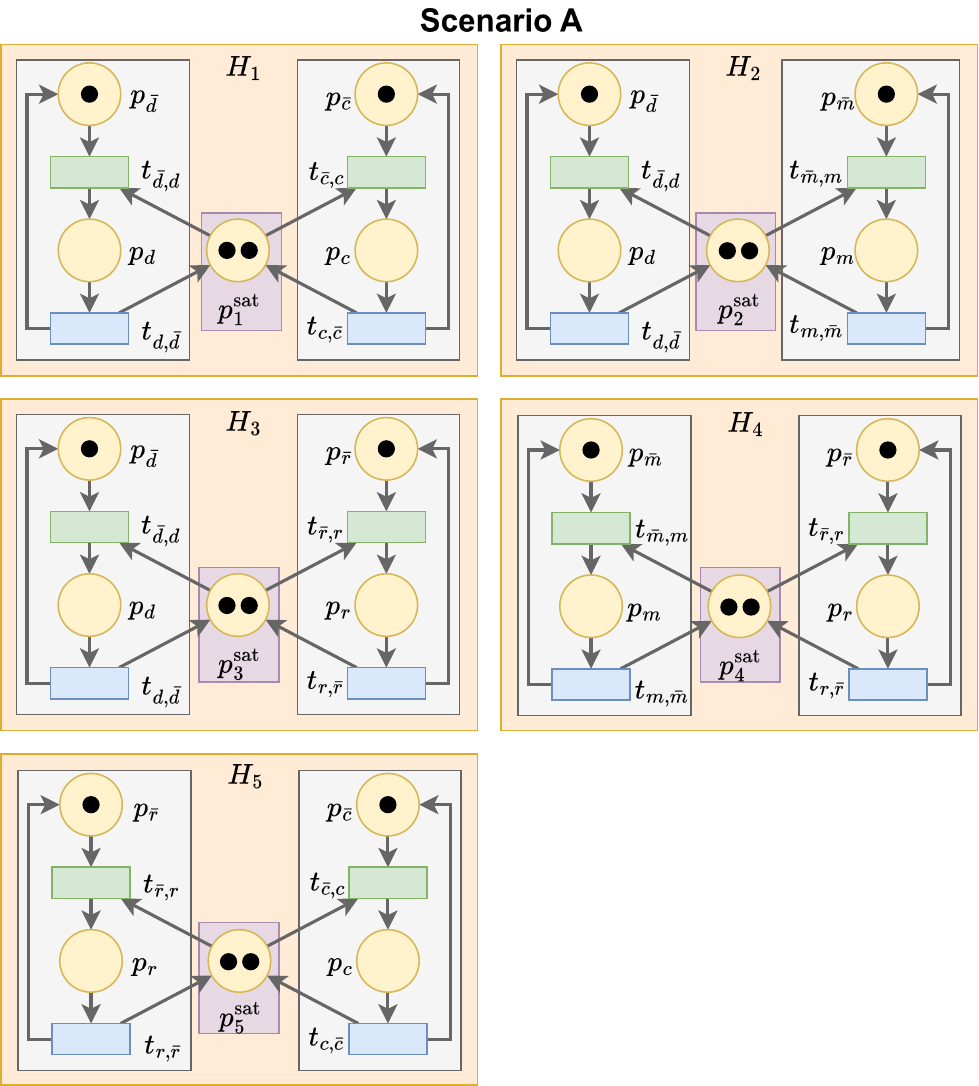}
	\caption{Clause-specific HyPN subnets $H_1,\dots,H_5$ corresponding to the five mutual-exclusion clauses in Scenario~A.}
	\label{fig:pn_scenario_a_step5}
\end{figure}

\textit{Clause-wise synthesis.}
The CNF is $\varphi^{CNF} = \varphi_1 \land \varphi_2 \land \varphi_3 \land \varphi_4 \land \varphi_5$, where $\varphi_1 = (\bar{d} \lor \bar{c})$, $\varphi_2 = (\bar{d} \lor \bar{m})$, $\varphi_3 = (\bar{d} \lor \bar{r})$, $\varphi_4 = (\bar{m} \lor \bar{r})$, and $\varphi_5 = (\bar{r} \lor \bar{c})$. Each clause $\varphi_i$ yields a HyPN subnet $H_i$ from the literal set $S_i \subseteq \mathcal{L}$. The resulting subnets are shown in Fig.~\ref{fig:pn_scenario_a_step5} and are constructed analogously.

\textit{Composition.}
The final HyPN is $H_A^{final} = H_{\{1,\ldots,5\}} = H_1 \oplus H_2 \oplus H_3 \oplus H_4 \oplus H_5$. Since no literal conflicts occur, composition is straightforward. Fig.~\ref{fig:pn_scenario_1_composition_from_1_to_5} shows the successive compositions $H_1 \oplus H_2$, $H_{\{1,2\}} \oplus H_3$, $H_{\{1,2,3\}} \oplus H_4$, and $H_{\{1,2,3,4\}} \oplus H_5$, leading to $H_{\{1,\ldots,5\}}$. This process does not induce exponential structural growth. In particular, all literal places already appear in the intermediate net $H_{\{1,2,3\}}$, and composing with $H_4$ and $H_5$ only adds two additional satisfaction places ($p^{sat}_4$ and $p^{sat}_5$) and the associated arcs. This highlights the structural scalability of the approach.

\begin{figure}
	\centering
	\includegraphics[width=1.0\linewidth]{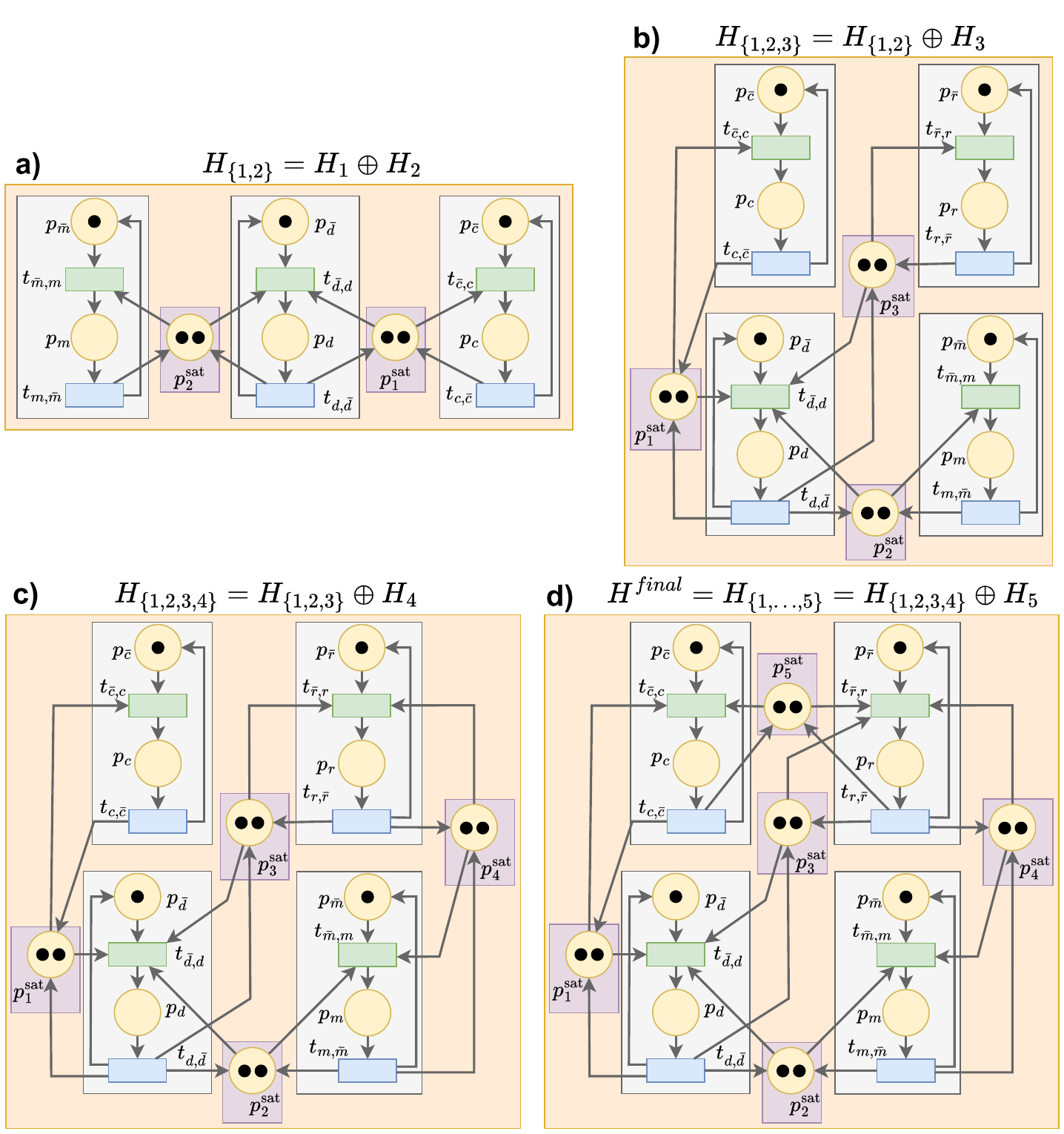}
	\caption{Stepwise composition of clause-specific HyPN subnets $H_1,\ldots,H_5$ in Scenario~A, resulting in $H^{final}_A = H_{\{1,\ldots,5\}}$.}
	\label{fig:pn_scenario_1_composition_from_1_to_5}
\end{figure}

\subsubsection{Observable behavior analysis}
\label{subsubsec:scenario-a-observable-behavior-analysis}

\textit{Satisfying valuations and observable markings.}
For $V=\{c,d,m,r\}$, only six of the $2^{|V|}=16$ Boolean valuations satisfy $\varphi$, in one-to-one correspondence with the observable markings of the synthesized HyPN (as shown in Table~\ref{tab:scenario-a-valuations}).
\begin{table}
	\centering
	\caption{Satisfying valuations $v_i \in Val(\varphi)$ and corresponding observable markings. The place ordering is $(p_{\bar{c}}, p_{c}, p_{\bar{r}}, p_{r}, p_{\bar{d}}, p_{d}, p_{\bar{m}}, p_{m}, p^{sat}_1, \ldots, p^{sat}_5)$.}
	\label{tab:scenario-a-valuations}
	\begin{tabular}{|c | |c | c | c | c || c |}
		\hline
		& \multicolumn{4}{c||}{Variables} & Observable marking \\ \hline
		Val. & $c$ & $r$ & $d$ & $m$ & $M_i = (\text{literal places}; \text{sat. places})$\\
		\hline \hline
		$v_1$ & 0 & 0 & 0 & 0 &  $M_0=(1,0,1,0,1,0,1,0; 2,2,2,2,2)$\\ \hline
		$v_2$ & 0 & 1 & 0 & 0 &  $M_5=(1,0,0,1,1,0,1,0; 2,2,1,1,1)$\\ \hline
		$v_3$ & 0 & 0 & 0 & 1 &  $M_3=(1,0,1,0,1,0,0,1; 2,1,2,1,2)$\\ \hline
		$v_4$ & 0 & 0 & 1 & 0 &  $M_4=(1,0,1,0,0,1,1,0; 1,1,1,2,2)$\\ \hline
		$v_5$ & 1 & 0 & 0 & 0 &  $M_1=(0,1,1,0,1,0,1,0; 1,2,2,2,1)$\\ \hline
		$v_6$ & 1 & 0 & 0 & 1 &  $M_2=(0,1,1,0,1,0,0,1; 1,1,2,1,1)$\\ \hline
	\end{tabular}
\end{table}

\begin{table*}
	\centering
	\caption{Scenario A: Firing sequences $\sigma_{i,j} \in \Sigma_H(M_i,M_j)$ for observable markings. The first column lists the active activities in each marking, the second the satisfying valuation of $\varphi$, and the third the observable state, where $S_i$ corresponds to $M_i$.}
	\label{tab:scenario_a_transitions_between_observable_markings}
	\begin{tabular}{|c|c|c|c||c|c|c|c|c|c|}
		\hline
		Active & Val. & State & Marking & $M_0$ & $M_1$ & $M_2$ & $M_3$ & $M_4$ & $M_5$ \\
		\hline \hline
		$\emptyset$ & $v_1$ & $S_0$ & $M_0$ &  -- & $\langle t_{\bar{c},c} \rangle$ & $\langle t_{\bar{c},c}, t_{\bar{m},m} \rangle$ & $\langle t_{\bar{m},m} \rangle$ & $\langle t_{\bar{d},d} \rangle$ & $\langle t_{\bar{r},r} \rangle$ \\ \hline
		$\{c\}$ & $v_5$ & $S_1$  & $M_1$ &  $\langle t_{c,\bar{c}} \rangle$ & -- & $\langle t_{\bar{m},m} \rangle$ & $\langle t_{\bar{m},m}, t_{c,\bar{c}} \rangle$ & $\langle t_{\bar{d},d}, t_{c,\bar{c}} \rangle$ & $\langle t_{\bar{r},r}, t_{c,\bar{c}} \rangle$ \\ \hline
		$\{c,m\}$ & $v_6$ & $S_2$ & $M_2$ &  $\langle t_{c,\bar{c}}, t_{m,\bar{m}} \rangle$ & $\langle t_{m,\bar{m}} \rangle$ & -- & $\langle t_{c,\bar{c}} \rangle$ & $\langle t_{\bar{d},d}, t_{c,\bar{c}}, t_{m,\bar{m}} \rangle$ & $\langle t_{\bar{r},r}, t_{c,\bar{c}}, t_{m,\bar{m}} \rangle$ \\ \hline
		 $\{m\}$ & $v_3$ & $S_3$& $M_3$ & $\langle t_{m,\bar{m}} \rangle$ & $\langle t_{\bar{c},c}, t_{m,\bar{m}} \rangle$ & $\langle t_{\bar{c},c} \rangle$ & -- & $\langle t_{\bar{d},d}, t_{m,\bar{m}} \rangle$ & $\langle t_{\bar{r},r}, t_{m,\bar{m}} \rangle$ \\ \hline
		$\{d\}$ & $v_4$ & $S_4$ & $M_4$ &  $\langle t_{d,\bar{d}} \rangle$ & $\langle t_{\bar{c},c}, t_{d,\bar{d}} \rangle$ & $\langle t_{\bar{c},c}, t_{\bar{m},m}, t_{d,\bar{d}} \rangle$ & $\langle t_{\bar{m},m}, t_{d,\bar{d}} \rangle$ & -- & $\langle t_{\bar{r},r}, t_{d,\bar{d}} \rangle$ \\ \hline
		$\{r\}$ & $v_2$ & $S_5$ & $M_5$ &  $\langle t_{r,\bar{r}} \rangle$ & $\langle t_{\bar{c},c}, t_{r,\bar{r}} \rangle$ & $\langle t_{\bar{c},c}, t_{\bar{m},m}, t_{r,\bar{r}} \rangle$ & $\langle t_{\bar{m},m}, t_{r,\bar{r}} \rangle$ & $\langle t_{\bar{d},d}, t_{r,\bar{r}} \rangle$ & -- \\ \hline
	\end{tabular}
\end{table*}

\textit{Transition sequences.}
Table~\ref{tab:scenario_a_transitions_between_observable_markings} summarizes representative firing sequences $\sigma_{i,j} \in \Sigma_H(M_i,M_j)$ between observable markings. Each entry $\langle t_1,\ldots,t_k \rangle$ denotes a firing sequence from $M_i$ to $M_j$, and the first column lists the active activities in each observable marking. The transitions directly correspond to activity activation and deactivation. For example, $\sigma_{0,1}=\langle t_{\bar{c},c}\rangle$ activates communication, whereas $\sigma_{2,4}=\langle t_{\bar{d},d}, t_{c,\bar{c}}, t_{m,\bar{m}} \rangle$ activates $d$ and deactivates $c$ and $m$.

Not all transitions between observable markings are atomic. For example, the transition from $M_1$ to $M_3$ may proceed as $\langle t_{\bar{m},m}, t_{c,\bar{c}} \rangle$, yielding $M_1 \xrightarrow{t_{\bar{m},m}} M_2 \xrightarrow{t_{c,\bar{c}}} M_3$, or as $\langle t_{c,\bar{c}}, t_{\bar{m},m} \rangle$, yielding $M_1 \xrightarrow{t_{c,\bar{c}}} M_0 \xrightarrow{t_{\bar{m},m}} M_3$. In both cases, the sequence passes through another observable marking, so it is non-atomic. This distinction becomes important below, when analyzing which transitions between observable states are direct and which require intermediate observable markings. When multiple admissible sequences exist, their selection may depend on the adopted policy, for example favoring shorter executions or those that preserve intermediate observable markings as safe stopping points.

\textit{Observable states.}
The observable states over $(c,r,d,m)$ are $S_0=(0,0,0,0)$, $S_1=(1,0,0,0)$, $S_2=(1,0,0,1)$, $S_3=(0,0,0,1)$, $S_4=(0,0,1,0)$, and $S_5=(0,1,0,0)$.

\textit{State-transition systems.} Fig.~\ref{fig:scenario_1_analysis_two_state_transition_systems} shows three transition relations over the observable states. Under reachability, a transition from $S_i$ to $S_j$ exists if there is a realizable sequence $\sigma \in T^*$ transforming $M_i$ into $M_j$, yielding a densely connected state-transition system (Fig.~\ref{fig:scenario_1_analysis_two_state_transition_systems}a). Under the HyPN transition semantics, a transition exists only if $\Sigma_H(M_i,M_j)\neq\emptyset$, i.e., if there is an atomic sequence from $M_i$ to $M_j$, so the induced relation is a subrelation of reachability, i.e., $\Rightarrow_H \subseteq \Rightarrow_R$ (Fig.~\ref{fig:scenario_1_analysis_two_state_transition_systems}b). 
Fig.~\ref{fig:scenario_1_analysis_two_state_transition_systems}c shows a further restriction, where direct atomic transitions are omitted whenever an alternative realizable path passes through intermediate observable states. This corresponds to a policy that prioritizes executions through states satisfying $\varphi$, for example to preserve safe stopping points or enforce a higher safety level. As a result, the induced transition graph becomes sparser, further reducing structural complexity. Although policy selection is outside the scope of this paper, it is a promising direction for future work.

\textit{Example and perspectives:}
Consider the transition from $S_5$ to $S_2$. Multiple realizable sequences exist:
a)~$S_5 \xrightarrow{\bar{r}} S_0 \xrightarrow{c} S_1 \xrightarrow{m} S_2$;
b)~$S_5 \xrightarrow{\bar{r}} S_0 \xrightarrow{m} S_3 \xrightarrow{c} S_2$;
c)~$S_5 \xrightarrow{c} (\text{internal}) \xrightarrow{m} (\text{internal}) \xrightarrow{\bar{r}} S_2$;
d)~$S_5 \xrightarrow{m} (\text{internal}) \xrightarrow{c} (\text{internal}) \xrightarrow{\bar{r}} S_2$.
Sequences (a--b) pass through observable states, whereas (c--d) remain internal. Hence, $\Sigma_H(M_5,M_2)\neq\emptyset$, so a direct transition appears in Fig.~\ref{fig:scenario_1_analysis_two_state_transition_systems}b.
In contrast, all realizable sequences from $S_0$ to $S_2$ pass through observable states, e.g., via $S_1$ or $S_3$, which implies $\Sigma_H(M_0,M_2)=\emptyset$. Accordingly, no direct transition appears between $S_0$ and $S_2$ in Fig.~\ref{fig:scenario_1_analysis_two_state_transition_systems}b.

These examples motivate the three perspectives in Fig.~\ref{fig:scenario_1_analysis_two_state_transition_systems}: reachability (a): based on admissible sequences; HyPN transition semantics (b): based on atomic sequences captured by $\Sigma_H(M_i,M_j)$; and a safety-oriented restriction (c): where some direct atomic transitions are omitted whenever an alternative admissible path passes through intermediate observable states. The last perspective reflects a policy that prioritizes paths through states satisfying $\varphi$, for example, to preserve intermediate observable states as safe stopping points. As a result, the induced graph becomes sparser, reducing structural complexity. While policy selection is outside the scope of this paper, this is an important direction for future work.

\begin{figure} 
	\centering 
	\includegraphics[width=0.9\linewidth]{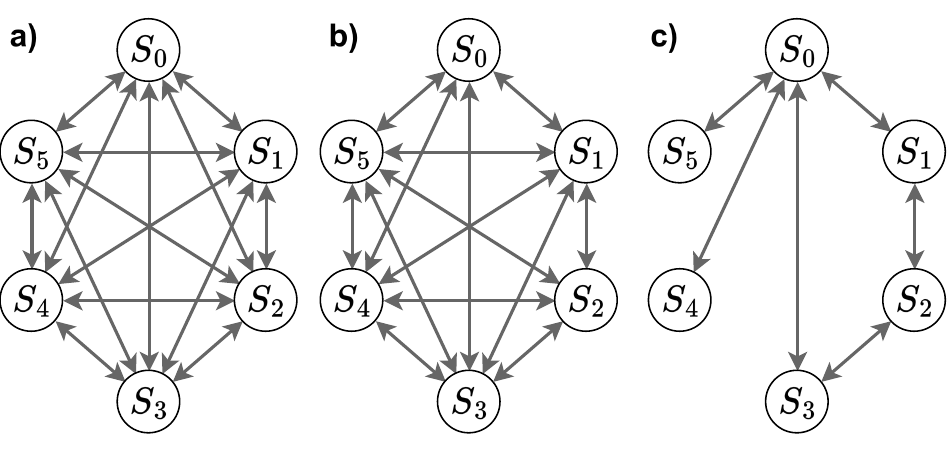}
	\caption{State-transition systems induced by the synthesized HyPN under three perspectives:
		(a) reachability, based on admissible sequences;
		(b) HyPN transition semantics, based on atomic sequences;
		(c) safety-oriented execution, omitting some direct atomic transitions to preserve intermediate observable states rather than bypassing them.}
	\label{fig:scenario_1_analysis_two_state_transition_systems}
\end{figure}

\subsection{Scenario B}
\label{subsec:scenario-B}

Let the supervisor control two activities: drive ($d$) and recharge ($r$). This scenario considers constraint-driven synthesis under resource availability, where action activation depends on the battery state of charge (SOC). The SOC is discretized into three mutually exclusive levels: $e_1 := (soc \leq \mathrm{MIN})$, $e_2 := (\mathrm{MIN} < soc \leq \mathrm{MAX})$, and $e_3 := (soc > \mathrm{MAX})$, where $soc$ denotes the battery state of charge. The constraints couple actions with resource levels and introduce exclusivity over $(e_1, e_2, e_3)$, which may lead to conflicting assignments on shared literal places during clause-level composition. Consequently, the synthesis requires the swap-based conflict resolution mechanism (Case~3 in Sec.~\ref{subsec:pn-composition}).

\subsubsection{HyPN synthesis}
\label{subsubsec:scenario-a-hypn-synthesis}

\textit{Constraints.} Consider the set of constraints $\mathcal{C} = \{C_1, C_2, C_3, C_4, C_5, C_6, C_7, C_8, C_9\}$, where:
\begin{enumerate}[leftmargin=1.0cm]
	\item[$C_1$:] $r \Rightarrow \bar{d}$: recharging excludes driving,
	\item[$C_2$:] $r \lor d$: at least one action is active,
	\item[$C_3$:] $d \Rightarrow (e_2 \lor e_3)$: driving requires sufficient energy,
	\item[$C_4$:] $e_1 \Rightarrow r$: low SOC implies recharging,
	\item[$C_5$:] $e_3 \Rightarrow \bar{r}$: high SOC disables recharging,
	\item[$C_6$:] $e_1 \Rightarrow (\bar{e}_2 \land \bar{e}_3)$,
	\item[$C_7$:] $e_2 \Rightarrow (\bar{e}_1 \land \bar{e}_3)$,
	\item[$C_8$:] $e_3 \Rightarrow (\bar{e}_1 \land \bar{e}_2)$,
	\item[$C_9$:] $e_1 \lor e_2 \lor e_3$: exactly one SOC level holds.
\end{enumerate}
These constraints define:
(i) exclusive choice between $r$ and $d$,
(ii) resource-dependent action feasibility, and
(iii) mutual exclusivity of SOC levels,
as captured by groups ($C_1$--$C_2$), ($C_3$--$C_5$), and ($C_6$--$C_9$), respectively.

\textit{Literal sets.} The Boolean variables are $V = \{r, d, e_1, e_2, e_3\}$, with $\mathcal{L} = \{x, \bar{x} \mid x \in V\}$.

\textit{Formula construction.}
The global constraint formula is defined as $\varphi = \bigwedge_{i=1}^{9} C_i$. Its CNF is $\varphi^{CNF} = (\bar{r} \lor \bar{d}) \land (r \lor d) \land (\bar{d} \lor e_2 \lor e_3) \land (\bar{e}_1 \lor r) \land (\bar{e}_3 \lor \bar{r}) \land (\bar{e}_1 \lor \bar{e}_2) \land (\bar{e}_1 \lor \bar{e}_3) \land (\bar{e}_2 \lor \bar{e}_3) \land (e_1 \lor e_2 \lor e_3)$.

\begin{figure}
	\centering
	\includegraphics[width=1.0\linewidth]{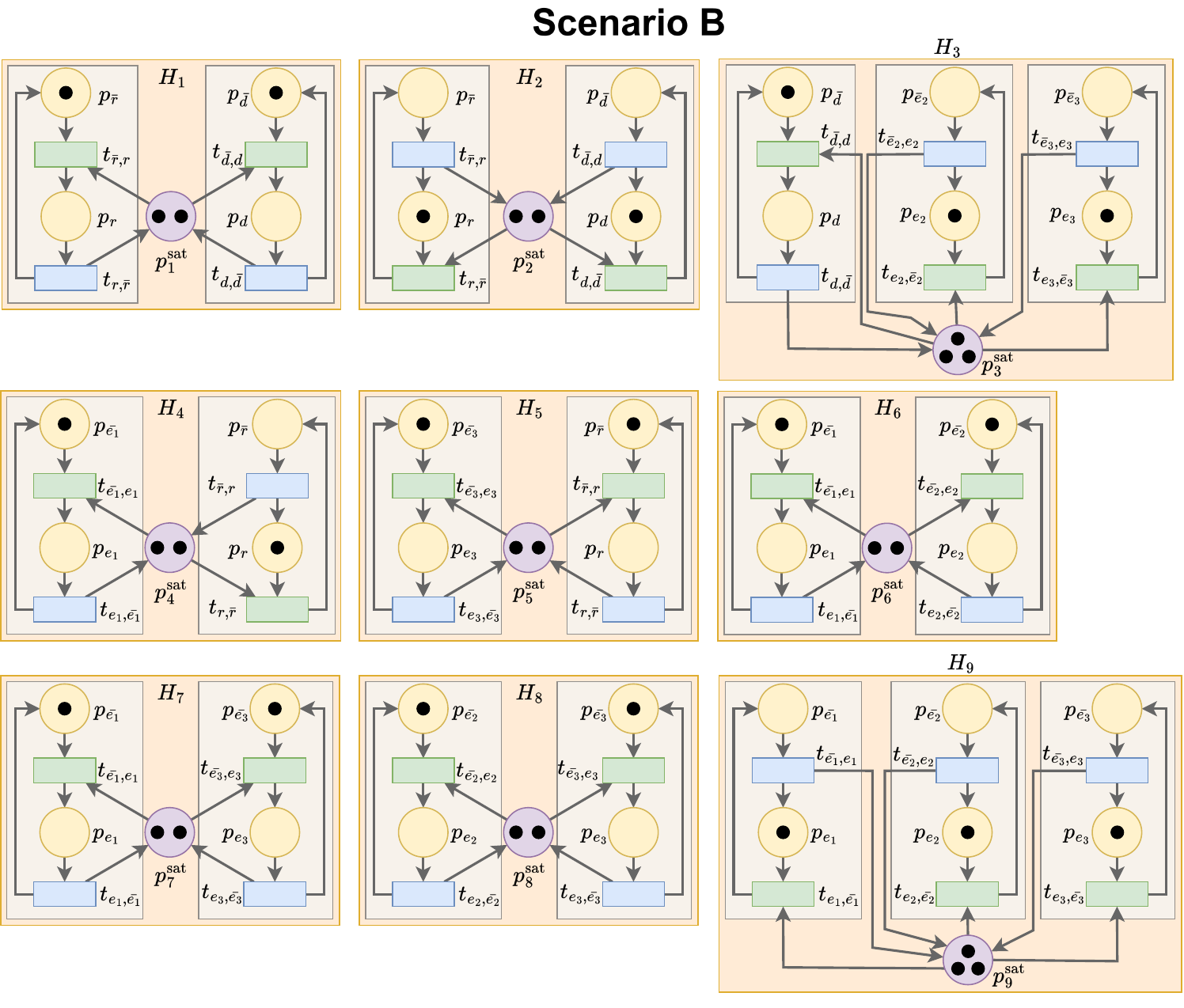}
	\caption{Synthesized sub-PNs: $H_1, \ldots, H_9$, corresponding to the sub-formulas $\varphi_1, \ldots, \varphi_9$ for Scenario~B.}
	\label{fig:pn_scenario_b_step5}
\end{figure}

\textit{Clause-wise synthesis.}
The CNF is $\varphi^{CNF} = \bigwedge_{i=1}^{9} \varphi_i$, where the clauses are:
$\varphi_1 = (\bar{r} \lor \bar{d})$, 
$\varphi_2 = (r \lor d)$,  
$\varphi_3 = (\bar{d} \lor e_2 \lor e_3)$, $\varphi_4 = (\bar{e}_1 \lor r)$,
$\varphi_5 = (\bar{e}_3 \lor \bar{r})$, 
$\varphi_6 = (\bar{e}_1 \lor \bar{e}_2)$,
$\varphi_7 = (\bar{e}_1 \lor \bar{e}_3)$, 
$\varphi_8 = (\bar{e}_2 \lor \bar{e}_3)$,	
$\varphi_9 = (e_1 \lor e_2 \lor e_3)$.
Each $\varphi_i$ yields a HyPN subnet $H_i$, resulting in nine subnets (Fig.~\ref{fig:pn_scenario_b_step5}).

\begin{figure}
	\centering
	\includegraphics[width=1.0\linewidth]{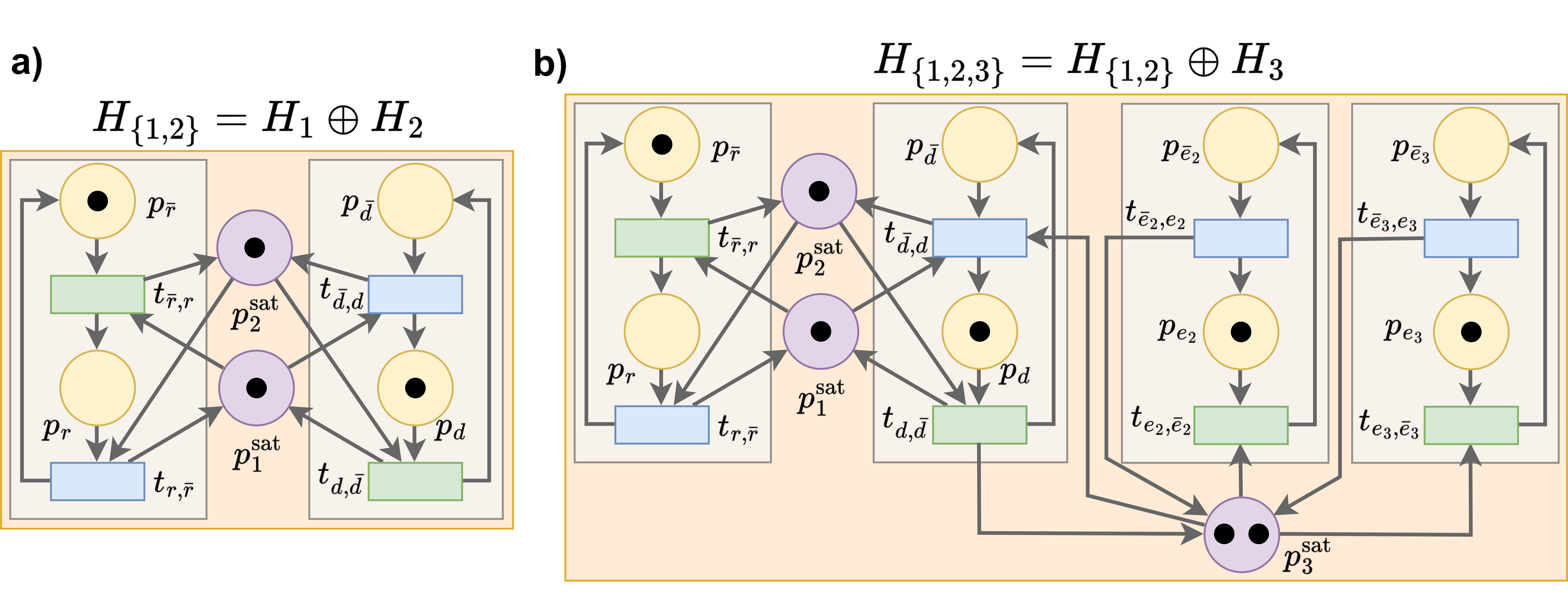}
	\caption{Representative intermediate compositions: a)~$H_{\{1,2\}} = H_1 \oplus H_2$, and b)~$H_{\{1,2,3\}} = H_{\{1,2\}} \oplus H_3$.}
	\label{fig:pn_scenario_b_composition_example}
\end{figure}

\textit{Composition.}
The final HyPN is $H^{final}_B = \bigoplus_{i=1}^{9} H_i$. The composition involves conflicting assignments on shared literal places, resolved via the swap-based procedure (Case~3 in Sec.~\ref{subsec:pn-composition}). Figure~\ref{fig:pn_scenario_b_composition_example} illustrates representative intermediate compositions. For example, composing $H_1 \oplus H_2$ yields conflicts in $r$ and $d$: the former is resolved by firing $t_{r,\bar{r}}$ in $H_2$, while the latter requires firing $t_{\bar{d},d}$ in $H_1$ to preserve satisfaction. Subsequent compositions proceed analogously. The final model $H^{final}_B = H_{\{1,\ldots,9\}}$ is shown in Fig.~\ref{fig:pn_scenario_b_final_pn_step_6_decomposed}. For clarity, the net is presented in five panels connected via satisfaction places, but should be interpreted as a single unified HyPN network, making input/output places and the influence of satisfaction places on variables more explicit.

\begin{figure}
	\centering
	\includegraphics[width=1.0\linewidth]{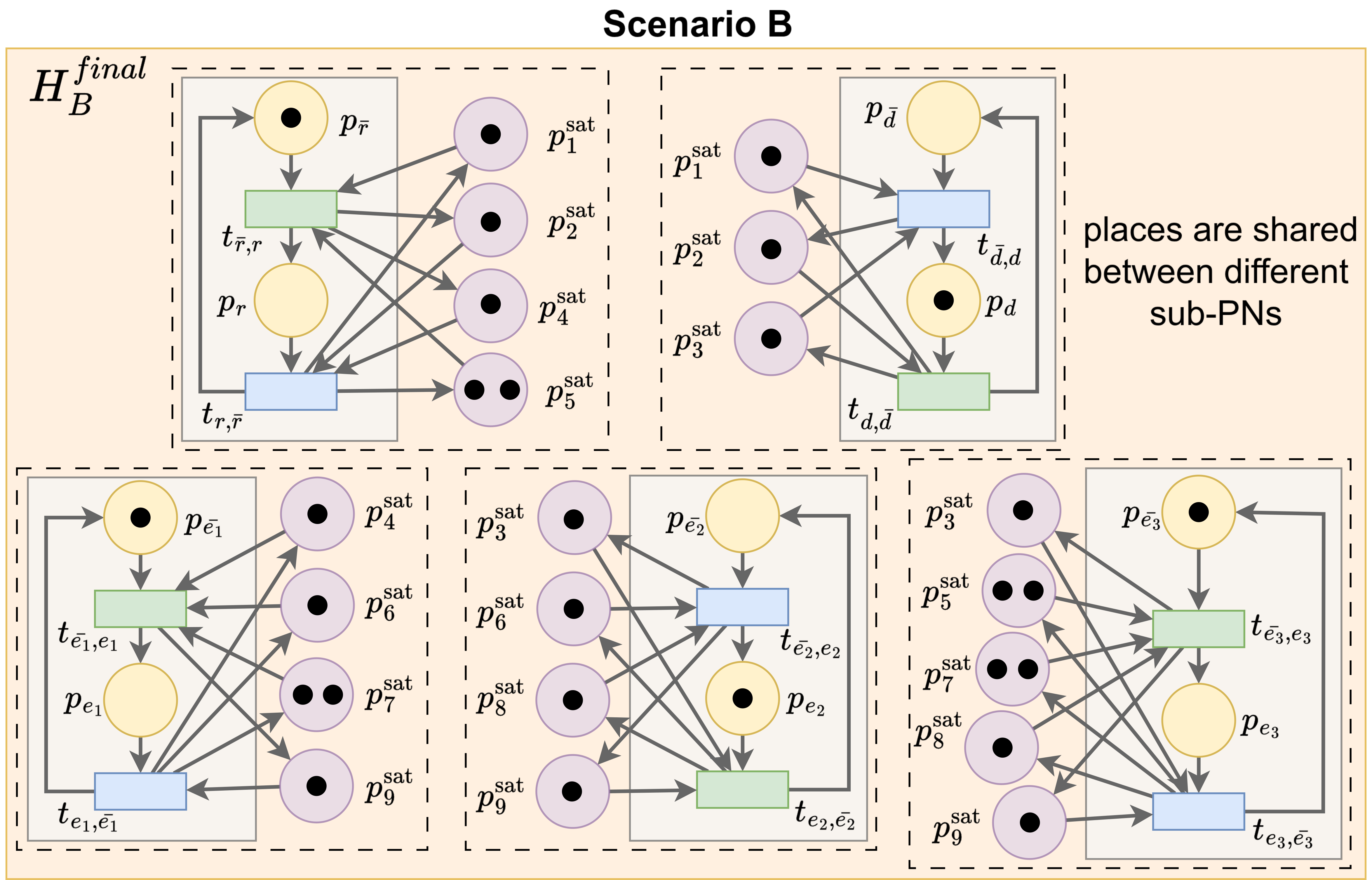}
	\caption{Final HyPN $H^{final}_B = H_{\{1,\ldots,9\}}$ for Scenario~B. The net is shown in five panels for clarity, interconnected via shared satisfaction places resulting from clause-level composition.}
	\label{fig:pn_scenario_b_final_pn_step_6_decomposed}
\end{figure}

\subsubsection{Observable behavior analysis}
\label{subsubsec:scenario-b-observable-behavior-analysis}

\textit{Satisfying valuations and observable markings.} The satisfying valuations of $\varphi$ and the corresponding observable markings $M_0,\ldots,M_3$ are given in Table~\ref{tab:scenario-b-valuations}. These define the observable states $S_0,\ldots,S_3$, with $|Val(\varphi)| = 4$. Each state $S_i$ corresponds to a valuation $v_{i+1}$.

\begin{table}
	\centering
	\footnotesize
	\setlength{\tabcolsep}{3pt}
	\caption{Satisfying valuations $v_i \in Val(\varphi)$ and corresponding observable markings. The place ordering is $(p_{\bar{r}}, p_{r}, p_{\bar{d}}, p_{d}, p_{\bar{e_1}}, p_{e_1}, p_{\bar{e_2}}, p_{e_2}, p_{\bar{e_3}}, p_{e_3}, p^{sat}_1,\ldots, p^{sat}_9)$.}
	\label{tab:scenario-b-valuations}
	\begin{tabular}{|c | |c | c | c | c | c || c |}
		\hline 
		& \multicolumn{5}{c||}{Variables} & Observable marking \\ \hline
		Val. & $r$ & $d$ & $e_1$ & $e_2$ & $e_3$ & $M_i = (\text{literal places}; \text{sat. places})$\\ \hline
		\hline
		$v_1$ & 0 & 1 & 0 & 1 & 0 &  $M_0=(1,\!0,\!0,\!1,\!1,\!0,\!0,\!1,\!1,\!0;\!1,\!1,\!1,\!1,\!2,\!1,\!2,\!1,\!1)$\\ \hline
		$v_2$ & 0 & 1 & 0 & 0 & 1 &  $M_1=(1,\!0,\!0,\!1,\!1,\!0,\!1,\!0,\!0,\!1;\!1,\!1,\!1,\!1,\!1,\!2,\!1,\!1,\!1)$\\ \hline
		$v_3$ & 1 & 0 & 1 & 0 & 0 &  $M_2=(0,\!1,\!1,\!0,\!0,\!1,\!1,\!0,\!1,\!0;\!1,\!1,\!1,\!1,\!1,\!1,\!1,\!2,\!1)$\\ \hline
		$v_4$ & 1 & 0 & 0 & 1 & 0 &  $M_3=(0,\!1,\!1,\!0,\!1,\!0,\!0,\!1,\!1,\!0;\!1,\!1,\!2,\!2,\!1,\!1,\!2,\!1,\!1)$\\ \hline
	\end{tabular}
\end{table}

\begin{table*}
	\centering
	\caption{Structure of transitions between observable markings (Scenario B).}
	\label{tab:scenario_b_transitions_between_observable_markings}
	\begin{tabular}{|c|c|c|c||c|c|c|c|}
		\hline
		Active & Val. & State & Marking & $M_0$ & $M_1$ & $M_2$ & $M_3$ \\
		\hline \hline
		$\{d, e_2\}$ & $v_1$ & $S_0$  & $M_0$ & -- & $\langle t_{e_2, \bar{e_2}}, t_{\bar{e_3}, e_3}  \rangle$  & $\langle t_{d,\bar{d}}, t_{e_2, \bar{e_2}}, t_{\bar{e_1},e_1}, t_{\bar{r}, r} \rangle$ & $\langle t_{d,\bar{d}}, t_{\bar{r}, r} \rangle$ \\ \hline
		$\{d, e_3\}$ & $v_2$ & $S_1$  & $M_1$ & $\langle t_{e_3, \bar{e_3}}, t_{\bar{e_2}, e_2} \rangle$ & -- & $\langle t_{d,\bar{d}}, t_{e_3, \bar{e_3}}, t_{\bar{e_1}, e_1}, t_{\bar{r}, r}  \rangle$ & $\langle t_{d,\bar{d}}, t_{e_3, \bar{e_3}}, t_{\bar{e_2}, e_2}, t_{\bar{r}, r} \rangle$ \\ \hline
		$\{r, e_1\}$ & $v_3$ & $S_2$  & $M_2$ & $\langle t_{r, \bar{r}}, t_{e_1, \bar{e_1}}, t_{\bar{e_2}, e_2}, t_{\bar{d}, d} \rangle$ & $\langle t_{r, \bar{r}}, t_{e_1, \bar{e_1}}, t_{\bar{e_3}, e_3}, t_{\bar{d}, d} \rangle$ & -- & $\langle t_{e_1, \bar{e_1}}, t_{\bar{e_2}, e_2} \rangle$ \\ \hline
		$\{r, e_2\}$ & $v_4$ & $S_3$  & $M_3$ & $\langle t_{r, \bar{r}},  t_{\bar{d}, d} \rangle$ & $\langle t_{r, \bar{r}}, t_{e_2, \bar{e_2}}, t_{\bar{e_3}, e_3}, t_{\bar{d}, d} \rangle$ & $\langle t_{e_2, \bar{e_2}}, t_{\bar{e_1}, e_1}  \rangle$ & -- \\ \hline
	\end{tabular}
\end{table*}

\textit{Transition structure.} The transitions between observable markings are summarized in Table~\ref{tab:scenario_b_transitions_between_observable_markings}. The induced state-transition system over $S_0,\ldots,S_3$ is fully connected under reachability, i.e., for any $(S_i,S_j)$ there exists a firing sequence $\sigma_{i,j}$ connecting $M_i$ and $M_j$.

\textit{Physical feasibility.} Not all logically admissible transitions are physically feasible, as some require intermediate observable states due to SOC dynamics. For example, a direct transition $S_2 \rightarrow S_1$ (from recharging at low SOC to driving at high SOC) would imply an instantaneous SOC increase and must proceed via $S_3$ (intermediate SOC level $e_2$). The physically feasible transitions are:
$S_1 \rightarrow S_0$ (SOC decreases during driving), 
$S_0 \rightarrow S_3$ (transition to recharging), 
$S_3 \rightarrow S_1$ (fully charged, driving resumes), 
$S_0 \rightarrow S_2$ (SOC drops to minimum), and 
$S_2 \rightarrow S_3$ (recharging). 
All remaining transitions violate the monotonic, resource-dependent SOC evolution.

\textit{Discussion and implications.} 
Logical feasibility does not imply physical realizability. While admissible transitions exist between all states, atomic transitions capture direct executability, and the underlying dynamics further restrict physically feasible transitions. Moreover, identical valuations (e.g., $e_2 = 1$ in $v_1$ and $v_4$) may correspond to different system behaviors (e.g., driving vs. recharging under different SOC levels), showing that system behavior is determined by transition sequences rather than individual states. Consequently, reasoning based solely on Boolean constraints may over-approximate feasible system evolution. This highlights a separation between logical, atomic, and physical levels, with HyPN capturing logical and atomic structure, while physical feasibility depends on system dynamics and is reflected in admissible transition sequences.

\subsection{Comparative summary}
\label{subsec:comparative-summary}

Table~\ref{tab:comparative-summary} summarizes key properties of the scenarios and highlights the general scaling behavior of the proposed synthesis procedure. The results confirm that the structural size grows linearly with the number of clauses, while the observable state space scales with the number of satisfying valuations. This highlights a separation between structural complexity and behavioral expressiveness, where a compact Petri net structure induces an exponentially large space of observable states.

\begin{table}
	\centering
	\caption{Comparison of synthesized HyPN models for two representative constraint classes}
	\scriptsize
	\label{tab:comparative-summary}
	\begin{tabular}{|c|p{1.5cm}|p{1.5cm}|c|}
		\hline
		Property & Scenario A & Scenario B & Scaling \\
		\hline
		Constraint class & mutual exclusion & resource-constrained & --- \\
		\hline
		Variables & 4 & 5 & $|V|$ \\
		\hline
		Clauses & 5 & 9 & $|C|$ \\
		\hline
		Places & 13 & 19 & $2|V| + |C|$ \\
		\hline
		Transitions & 8 & 10 & $2^{|V|}$ \\
		\hline
		Size of state space & 16 & 32 & $2^{|V|}$ \\
		\hline 
		Observable states & 6 & 4 & $|\mathrm{Val}(\varphi)|$ \\
		\hline
		Max observable edges & 12 & 5 & problem-dependent \\
		\hline
		Composition & direct & requires swaps & conflict-dependent \\
		\hline
	\end{tabular}
\end{table}					

\section{Discussion}
\label{sec:discussion}

The proposed HyPN model separates logical correctness from execution semantics by distinguishing observable markings from the underlying transition structure. This enables reasoning about which markings satisfy constraints and how they can be reached through admissible execution sequences.

\textbf{Structural scalability and state-space explosion.}
A key advantage of the synthesis is that it avoids explicit transition structures between all admissible Boolean states. Direct transitions between all admissible states induces a complete directed graph with $(2^{|V|} - 1)(2^{|V|} - 2)$ transitions. In contrast, the synthesized HyPN has linear structural complexity in $|V|$ and $|C|$, with exactly $2|V|$ transitions and $2|V| + |C|$ places. This separates structural complexity from behavioral complexity: while the observable marking space may be exponential, the execution model remains compact and compositional.

\textbf{From modeling to control and decision support.}
The framework supports constraint-aware control by enabling selection of actions that guarantee constraint satisfaction. This is relevant for safety-critical systems, where maintaining admissibility is essential. Explicit representation of admissible and inadmissible transitions supports AI explainability.

\textbf{Execution semantics and admissibility.}
The analysis of Scenarios~A and~B shows that reachability alone is insufficient to characterize system behavior. While logical constraints define the admissible state space, execution semantics determines feasible transitions. Not all logically admissible transitions are physically realizable. This requires joint reasoning about state validity and execution feasibility.

\textbf{Comparison with existing approaches.}
Compared to temporal-logic-based synthesis, the proposed approach prioritizes scalability over  expressiveness, focusing on invariant properties while avoiding state-space explosion. Compared to supervisory control of PNs, which modifies an existing model, it directly synthesizes a model from constraints, enabling a more systematic design. Table~\ref{tab:method-comparison} summarizes key differences. In particular, HyPN achieves linear structural scaling and direct synthesis from constraints, unlike the higher complexity of temporal-logic-based synthesis and the model dependency of supervisory control.

\textbf{Limitations and future work.}
For constraint-driven synthesis, the set of observable markings $\mathcal{M}^{obs}$ is directly determined by the constraint structure and can be evaluated in $O(|P^{sat}|)$. However, the framework does not yet provide mechanisms for selecting preferred execution sequences or enforcing optimality, and is limited to Boolean constraints and invariant properties, without support for temporal or probabilistic aspects. Future work will focus on policy selection for admissible transition sequences and integration with temporal reasoning, enabling direct evaluation of temporal logic over execution sequences without explicit automata.

\textbf{Implications for system design.}
The proposed framework provides a foundation for designing constraint-aware systems, where correctness is guaranteed by construction while maintaining a compact and scalable execution model. This enables integration of constraint satisfaction, execution reasoning, and decision-making within a unified formal representation for safety-critical and autonomous systems.

\begin{table}
	\setlength{\tabcolsep}{3pt}
	\centering
	\caption{Comparison of HyPN with synthesis and control}
	\scriptsize
	\label{tab:method-comparison}
	\begin{tabular}{|p{1.6cm}|p{1.2cm}|p{1.1cm}|p{2.1cm}|p{1.7cm}|}
		\hline
		Approach & Input & Scalability & Main strength & Limitation \\
		\hline
		Temporal-logic synthesis & temporal logic & exponential & expressive specs & high complexity \\
		\hline
		PN supervisory control & PN + constraints & model-dependent & enforcement on existing PN & requires a predefined model \\
		\hline
		HyPN & Boolean constraints & linear in $|V|, |C|$ & direct synthesis with explicit semantics & no temporal or optimal policies\\
		\hline
	\end{tabular}
\end{table}					

\section{Conclusions}
\label{sec:conclusions}
This paper introduced the HyPN model, which separates observable markings from the underlying PN execution and defines explicit execution semantics on them. A constraint-driven synthesis procedure was developed to automatically construct HyPN models from Boolean specifications in conjunctive normal form, guaranteeing that all observable markings satisfy the constraints by construction.

The proposed execution semantics, based on admissible firing sequences, highlights a fundamental gap between logical constraint satisfaction and realizable system behavior. The presented scenarios demonstrate that not all logically admissible transitions are executable, highlighting the need to distinguish between logical and executable behavior. Overall, the framework shifts system design from post-hoc verification toward constructing models that guarantee correctness by design, while maintaining a compact and scalable execution structure. Future work will focus on selection and control of admissible transition sequences, and extension of the framework with recovery mechanisms enabling transitions from non-observable to observable states, supporting initialization and fault recovery scenarios.			

\section*{Acknowledgments}
This work was supported by the Polish National Agency for Academic Exchange (NAWA) under the Bekker Programme, grant BPN/BEK/2024/1/00079, during a research stay at the Jet Propulsion Laboratory, California Institute of Technology.

Part of this work was carried out at the Jet Propulsion Laboratory, California Institute of Technology, under a contract with the National Aeronautics and Space Administration.

Part of this work was further developed at the Warsaw University of Technology.

\end{document}